\documentclass[nonacm,sigplan]{acmart}

\usepackage{tikz}
\usepackage{amsmath}
\usepackage{mathtools}
\usepackage{mathpartir}
\usepackage{caption}
\usepackage{subcaption}
\usepackage{listings}
\usepackage{multirow}
\usepackage[ruled,linesnumbered]{algorithm2e}
\usepackage{threeparttable}
\usepackage{pifont}
\usepackage{bbding}
\usepackage{pdfpages}
\usepackage{xspace}

\usepackage{hyperref}
\hypersetup{
  colorlinks   = true,    
  urlcolor     = olive,    
  linkcolor    = red,    
  citecolor    = red      
}

\newcommand{\parh}[1]{\noindent\textbf{#1}}

\newcommand{\F}{Fig.}

\newcommand{\T}{Table}
\renewcommand{\S}{Sec.}

\newcommand{\tool}{\textsc{CipherGuard}\xspace}
\newcommand{\ztool}{\textsc{Zebrafix}\xspace}
\newcommand{\otool}{\textsc{Obelix}\xspace}
\newcommand{\ftool}{\textsc{CipherFix}\xspace}
\newcommand{\htool}{\textsc{CipherH}\xspace}

\newtheorem{theorem}{Theorem}[section]
\newtheorem{lemma}[theorem]{Lemma}

\begin{document}

\title{CipherGuard: Compiler-aided Mitigation against Ciphertext Side-channel Attacks}

\author{Ke Jiang}
\affiliation{
  \institution{Southern University of Science and Technology}
  \city{}
  \country{}}
\email{ke006@e.ntu.edu.sg}

\author{Sen Deng}
\affiliation{
  \institution{Hong Kong University of Science and Technology}
  \city{}
  \country{}}
\email{sdengan@cse.ust.hk}

\author{Yinshuai Li}
\affiliation{
  \institution{Southern University of Science and Technology}
  \city{}
  \country{}}
\email{12231139@mail.sustech.edu.cn}

\author{Shuai Wang}
\affiliation{
  \institution{Hong Kong University of Science and Technology}
  \city{}
  \country{}}
\email{shuaiw@cse.ust.hk}

\author{Tianwei Zhang}
\affiliation{
  \institution{Nanyang Technological University}
  \city{}
  \country{}}
\email{tianwei.zhang@ntu.edu.sg}

\author{Yinqian Zhang}
\affiliation{
  \institution{Southern University of Science and Technology}
  \city{}
  \country{}}
\email{yinqianz@acm.org}

\begin{abstract}
Recently, the new ciphertext side channels resulting from the deterministic memory encryption in Trusted Execution Environments (TEEs), enable ciphertexts to manifest identifiable patterns when being sequentially written to the same memory address.
Attackers with read access to encrypted memory in TEEs can potentially deduce plaintexts by analyzing these changing ciphertext patterns. 
In this paper, we design \tool, a compiler-based mitigation tool to counteract ciphertext side channels with high efficiency and security guarantees.
\tool\ is based on the LLVM ecosystem, and encompasses multiple defense strategies, including software-assisted probabilistic encryption, secret-aware register allocation, and diversion-based obfuscation.
The design of \tool\ demonstrates that compiler techniques are highly effective for fine-grained control over mitigation code generation and assisted component management.
Through a comprehensive evaluation, it demonstrates that \tool\ can strengthen the security of various cryptographic implementations more efficiently than existing state-of-the-art defense, i.e., \ftool.
In its most efficient strategy, \tool\ incurs an average performance overhead of only 1.41$\times$, with a maximum of 1.95$\times$.
\end{abstract}


\maketitle
\pagestyle{plain}

\section{Introduction}
\label{sec:introduction}

Data security in cloud computing has become a major concern, hindering data owners from fully leveraging Virtual Machines (VMs) offered by public service providers.
To address this, major processor vendors introduce a hardware-based technology known as TEE, such as AMD SEV~\cite{kaplan2016amd}, SEV-ES~\cite{kaplan2017protecting} and SEV-SNP~\cite{kaplan2020amd}, Intel SGX~\cite{intelsgx, johnson2021supporting} and TDX~\cite{inteltdx} and ARM CCA~\cite{armcca}, which provides an isolated environment with memory encryption to fortify the integrity and confidentiality of VMs against privileged attackers, such as malicious hypervisors or host OS.
However, recent works disclosed a new vulnerability in SEV-SNP, namely ciphertext side-channel attack~\cite{li2021cipherleaks, li2022systematic}, where deterministic memory encryption enables attackers to construct a one-to-one mapping between plaintext and ciphertext by continuously monitoring changes in encrypted guest memory, ultimately recovering the plaintext.
More seriously, this attack is not limited to SEV-SNP but can threaten any TEE processors utilizing deterministic memory encryption with the memory bus snooping technique~\cite{lee2020off}.

Several studies have aimed at addressing this threat~\cite{deng2023cipherh, dfsan, wichelmann2023cipherfix}.
For example, \htool~\cite{deng2023cipherh} detects ciphertext side-channel vulnerabilities in popular cryptography libraries.
This detection work tracks sensitive memory writes by DFSan~\cite{dfsan} and conducts symbolic execution with constraint solving to identify potential vulnerabilities.
\ftool~\cite{wichelmann2023cipherfix} mitigates such leakage through static binary instrumentation, transforming each sensitive memory write into a copy with masking operations in a new section.
More recently, \otool~\cite{wichelmann2024obelix} leverages ORAM to protect code blocks and small data regions, while applying interleaving (padding) for larger data regions to ensure fresh ciphertexts.

Despite these efforts, significant shortcomings remain. 
Automated mitigation approaches such as \ftool\ and \otool\ incur \textbf{significant performance overheads}, making them impractical for real-world deployment. 
\ftool\ incurs overhead from 2.4$\times$ to 16.8$\times$ even in its most efficient variant, while \otool\ suffers from the inherently high overhead of ORAM, causing slowdowns of several hundred times.
On the other hand, detection-based tools like \htool\ leave developers with hundreds of alerts to \textbf{triage and patch manually}, a process that is error-prone, time-consuming, and offers no guarantee that fixes will not introduce new vulnerabilities. 
These limitations underscore the inadequacy of existing solutions and the urgent need for a practical, efficient, and developer-friendly mitigation framework.

Intuitively, integrating mitigation code during the compilation process enables more efficient defenses.
The most critical advantage of compiler-based repair lies in its global program visibility, which enables fine-grained resource management and the automated embedding of security strategies into mitigation code while preserving performance.
Building on this insight, we propose \tool, a compiler-based mitigation tool that automatically detects and protects sensitive memory writes during compilation to counteract ciphertext side channels. 
\tool\ performs in-place code insertion, preserving the original control flow and minimizing disruption to the code structure.
With compiler support, \tool\ flexibly manages resources such as registers and memory layouts, thereby optimizing the performance of mitigation code.
More importantly, \tool mitigates vulnerabilities across all relevant branches by leveraging the global perspective of the compiler, including propagation of data flow and control flow, preventing residual issues that could result from localized repairs.
In terms of security guarantees, \tool\ employs multi-strategy defenses: (1) software-assisted probabilistic encryption that randomizes secrets with nonces upon memory writes by masking, (2) secret-aware register allocation that retains secrets in registers to reduce memory exposure, and (3) diversion-based obfuscation that inserts decoy values into secret locations. 
By implementing these strategies, \tool\ achieves stronger security guarantees against ciphertext side channels.
Therefore, these advantages position \tool\ as an effective solution for achieving both high performance and strong security in ciphertext side-channel mitigation.

When implementing proposed multiple strategies, \tool\ is equipped with streamlined random nonce management, flexible register allocation, precise handling of variable lengths, reduced memory operations through obfuscation, and numerous compiler-driven resource optimizations, thereby achieving a satisfactory performance overhead.
In its most efficient strategy, \tool\ incurs an average performance overhead of only 1.41$\times$, with a maximum of 1.95$\times$.
This is a significant performance improvement over \ftool, highlighting the efficacy of the compiler-based methodology in defeating ciphertext side channels.
In sum, this paper makes the following contributions:

\begin{itemize}
\item It proposes multi-strategy defenses with enhanced security guarantees into a single compiler-based mitigation process, applying each strategy to its most suitable object.

\item It exhibits performance efficiency by leveraging the compiler’s inherent capabilities to preserve control flow, manage resources, and benefit from optimized nonce handling, flexible register allocation, and reduced memory operations when applying multi-strategy defenses.

\item It implements \tool, an LLVM-based mitigation tool for defending against ciphertext side channels. Evaluation on real-world cryptographic software set demonstrates that \tool\ secures all sensitive memory writes with superior performance compared to \ftool.
\end{itemize}


\section{Background and Related Work}
\label{sec:background}

\subsection{Ciphertext Side-Channel Attacks}
\label{subsec:ciphertext}

The ciphertext side channel originates from the deterministic memory encryption implemented in AMD's TEE.
The encrypted memory is calculated by an XOR-Encrypt-XOR (XEX) mode, expressed as: $c = ENC(m \oplus T(P_{m})) \oplus T(P_{m})$, where the plaintext $m$ undergoes the XOR operations before and after AES-128 encryption with a tweak value $T(P_{m})$ that incorporates the physical address $P_{m}$.
Without freshness, encryption of the same plaintext at a given physical address produces the identical ciphertext.
Note that this vulnerability extends to other deterministic encryption-based TEE architectures as long as attackers have read accesses to ciphertext (via software access~\cite{li2021cipherleaks} or memory bus snooping~\cite{lee2020off}).
Two attack schemes are introduced in work~\cite{li2022systematic}.

\begin{figure}[htbp]
\centering
\includegraphics[width=0.95\linewidth]{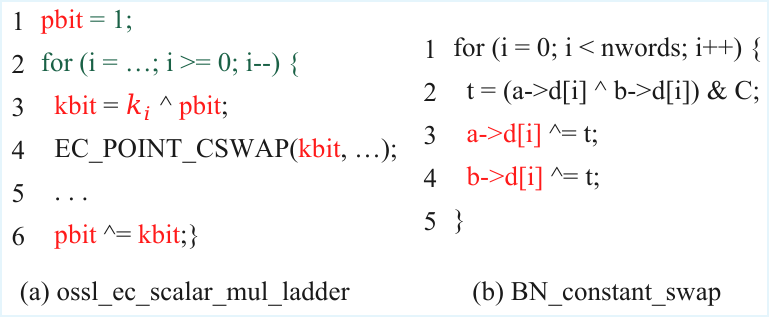}
\caption{Ciphertext side-channel examples.}
\label{fig:channel}
\end{figure}


\parh{The \textit{Dictionary} attack} involves continuous monitoring the ciphertext at a fixed address to construct a dictionary containing mappings of ciphertext-plaintext pairs.
Consider the code snippet shown in \F~\ref{fig:channel}(a), extracted from the ECDSA Montgomery ladder algorithm implemented in OpenSSL-3.0.2.
In each loop iteration, the \texttt{BN\_is\_bit\_set} function (denoted by $k_{i}$ in line 3) is utilized to obtain one bit of the secret $k$.
Following this, the $kbit$ variable is computed through an XOR operation with the value in $pbit$, which is then written back to $pbit$.
Given the dual XOR operations in lines 3 and 6, $pbit$ ultimately stores each bit of the secret $k$.
The attacker records consecutive ciphertext pairs ($pbit$-$kbit$) both before and after the \texttt{BN\_is\_bit\_set} function, aiming to deduce $k_{i}$ in each iteration based on the changes observed in ciphertext pairs.


\parh{The \textit{Collision} attack}, in contrast, focuses on identifying repetitions or alterations in certain ciphertexts to break the constant-time mechanism.
\F~\ref{fig:channel}(b) shows the constant-time-swap function \texttt{BN\_constant\_swap}.
This function takes two variables $a$ and $b$, along with a decision $C$ (e.g., $kbit$ in line 4 of \F~\ref{fig:channel}(a)).
If $C$ is set to 1, the values of $a$ and $b$ are exchanged, leading to observable changes in the ciphertext. Conversely, if $C$ is 0, the ciphertext remains unaltered.
In this way, the \textit{Collision} attack recovers the decision $C$, undermining the constant-time component.

Currently, many well-known cryptographic applications are vulnerable to this attack, including RSA and ECDSA (like \textit{secp256k1} and \textit{secp384r1}) equipped with constant-time algorithms, ECDSA from WolfSSL-5.3.0, ECDSA and RSA from MbedTLS-3.1.0, and EdDSA (\textit{Ed25519}) from OpenSSH adopted by Ubuntu LTS 20.04~\cite{li2021cipherleaks, li2022systematic}.

\subsection{Existing Countermeasures}
\label{subsec:countermeasures}

Hardware-based countermeasures provide stronger security by eliminating ciphertext side channels, but they require extensive validation before chip manufacturing. In contrast, we choose a software-based approach, enabling quicker implementation and deployment without modifying hardware.
Unfortunately, existing countermeasures for cache and timing side channels~\cite{percival2005cache, osvik2006cache, zhang2012cross, yarom2014flush, liu2015last, yarom2014recovering, ryan2019return, aranha2020ladderleak}, like constant-time cryptography, cannot mitigate ciphertext side channels. 
While constant-time cryptography avoids secret-dependent branches and memory accesses, it has been shown to be ineffective against ciphertext side-channel attacks~\cite{li2021cipherleaks, li2022systematic, deng2023cipherh}.


Without detailed implementation, AMD's whitepaper~\cite{amdmeasures} and Li et al.~\cite{li2022systematic} proposed countermeasures as follows.

\begin{itemize}
\item[1)] Preserving secret variables in registers instead of memory enhances security~\cite{li2022systematic}, but faces implementation challenges due to limited general purpose registers.

\item[2)] Avoiding the reuse of fixed memory addresses ensures fresh ciphertexts~\cite{li2022systematic, amdmeasures}, but requires extra memory and precise runtime reference management.

\item[3)] Introducing a random nonce to the plaintext with each memory write increases ciphertext unpredictability~\cite{li2022systematic}. This includes masking and padding strategies~\cite{amdmeasures}.
\end{itemize}

Translating these conceptual insights into a practical mitigation implementation is far from trivial, as it requires overcoming significant challenges to preserve both functionality and security. 
This reveals a critical research gap, which our work addresses by developing a principled, compiler-based methodology for mitigating ciphertext side channels. 
The difficulty underscores both the theoretical novelty and the practical value of our contribution.

\section{Overview}
\label{sec:overview}

\subsection{Threat Model}
\label{subsec:threat}

We adopt the same threat model as prior works on ciphertext side channels~\cite{li2021cipherleaks,li2022systematic}, where a privileged adversary aims to steal secrets from cryptographic programs running in a confidential VM (e.g., AMD SEV-SNP). 
The adversary has full system privileges and can read the ciphertext of the encrypted memory~\cite{li2022systematic}. 
We consider a powerful adversary capable of single-step attacks~\cite{wilke2023sev}: by controlling processes within a confidential VM and pausing after each instruction, the adversary can precisely identify optimal observation points in the control flow. 
We mainly focus on the protection of cryptographic software against ciphertext side-channel attacks, while assuming the OS kernel is patched to prevent register leakage during context switch~\cite{li2022systematic}.
For the recent relocation-based attacks~\cite{yany:sec2025:relocate_vote, heraclesCPA2025}, we assume that these attacks have been fixed, such that the hypervisor’s ability to move or swap guest pages is disabled.

\begin{figure}[htbp]
\centering
\includegraphics[width=0.95\linewidth]{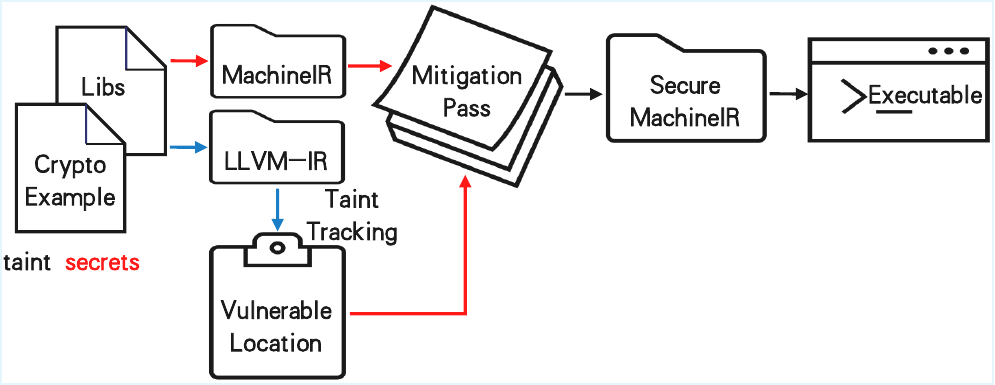}
\caption{Workflow of \tool.}
\label{fig:workflow}
\end{figure}

\subsection{Design Choices of \tool}
\label{subsec:workflow}

In designing \tool, we take the following considerations into account.

\noindent \textbf{Workflow.}
\tool is an LLVM-based methodology that mitigates ciphertext side channels \textit{at the compilation stage}.
\F~\ref{fig:workflow} illustrates the workflow of \tool, which consists of two phases: dynamic taint analysis and static rewriting.
First, \tool\ captures all sensitive memory writes and their precise memory references at IR level through dynamic taint analysis (the blue arrow pointing to the taint tracking component).
By seamlessly aligning the taint information with the Machine-IR level, \tool\ enables the application of multiple mitigation strategies during the \textit{recompilation} of target programs, as indicated by the red arrow pointing to the process section.
The modified Machine-IR is then compiled into a secure executable that can run directly within existing TEEs.

\noindent \textbf{Application Scope.}
The main audience for \tool\ is developers looking to deploy cryptographic software on modern TEEs. 
We highlight that \tool\ can mitigate standalone cryptographic software, meaning that a third-party library can be compiled into a complete cryptographic application.
Furthermore, both the target cryptographic software and the third-party libraries are taken from their C code.

\noindent \textbf{Non–Source-level Fixes.}
We refrain from applying ciphertext side-channel mitigation at the source-code level for two main reasons. 
First, ciphertext side channels are inherently tied to memory writes, and such sensitive operations can only be precisely identified in the compiler backend. 
Second, source-level fixes are often invalidated by compiler optimizations.
For instance, the register spilling-reloading mechanism may introduce additional memory writes in order to conserve register resources. 
Therefore, \tool\ implements mitigation in the compiler backend rather than at the source level.

\begin{figure*}[ht]
\centering
\includegraphics[width=0.95\textwidth]{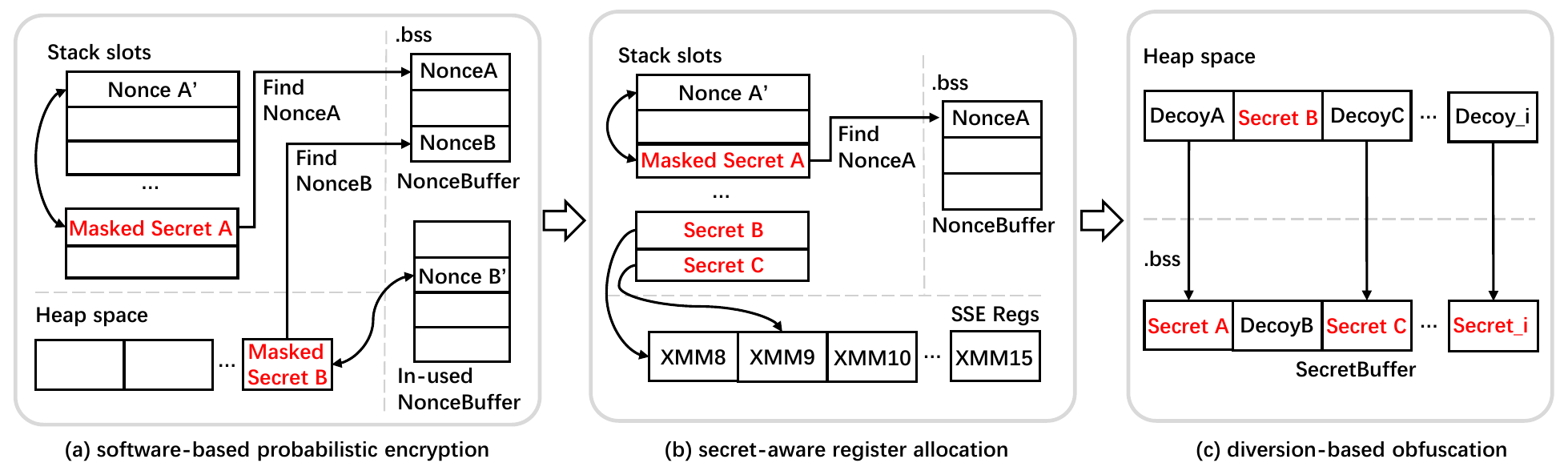}
\caption{Three mitigation strategies proposed in \tool. In (a), \texttt{rdrand} is used to pre-generate nonce buffer, while omitting the AES-NI and XorShift128+ variants for generating a random nonce on-the-fly.}
\label{fig:variants}
\end{figure*}

\noindent \textbf{Dynamic Taint Analysis.}
We adopt dynamic taint analysis primarily for its scalability and practical effectiveness. 
Unlike static approaches, which current research has shown to be limited to analyzing only small program fragments~\cite{wang2019identifying, brotzman2019casym}, dynamic taint analysis can handle large and complex software systems. 
Moreover, prior work on side-channel analysis, including \htool, \ftool, and cache-based studies~\cite{wang2017cached, jiang2022cache}, demonstrates that program execution paths often exhibit strong similarity. 
This observation suggests that dynamic taint analysis can effectively capture the majority of secret-dependent flows, making it a more suitable choice.

\noindent \textbf{IR and Machine-IR.}
Tainting during compilation is commonly used to detect program vulnerabilities. 
Tools like DFSan have been employed in prior work~\cite{borrello2021constantine, deng2023cipherh} to track secret-related IR instructions. 
Performing taint analysis on the final optimized IR eliminates the influence of front-end compiler optimizations. 
Further, mitigation at the Machine-IR level is necessary because backend optimizations, such as register allocation, can inadvertently compromise ciphertext side-channel defenses. 
To ensure effectiveness, mitigation must be applied at the LLVM Machine-IR level within register allocation pass, aligning taint information from the IR with the Machine-IR. 
Although IR instructions are replaced by target instructions, their semantics are largely preserved through instruction selection and emission.

After register allocation is completed, the instrumentation in the Machine-IR becomes stable and can be safely lowered to the binary, ensuring that the inserted mitigation code is not disrupted by subsequent compilation passes. 
This stability arises because, all virtual registers have been mapped to physical registers, and most remaining passes are limited to instruction lowering and scheduling rather than transformations that could remove or reorder the instrumentation.

\subsection{Mitigation Strategies and Challenges}
\label{subsec:challenges}

\tool\ proposes three mitigation strategies to resist ciphertext side-channel vulnerabilities, as illustrated in~\F~\ref{fig:variants}.

\parh{Strategy 1: Software-assisted Probabilistic Encryption} employs hardware instructions to generate nonces for masking sensitive memory stores, as illustrated in~\F~\ref{fig:variants}(a).
This strategy includes three variants: 1) using the \texttt{rdrand} instruction, optimized with a pre-generated nonce buffer; 2) leveraging the \texttt{vaesenc} instruction from AES-NI; and 3) applying a shift/rotate-based linear transformation derived from XorShift128+~\cite{vigna2017further}. 
Two random nonce buffers are essential to ensure secure masking of sensitive data: one buffer supplies the initial nonces used to mask sensitive variables upon their first store, while the other maintains the active nonces currently associated with each variable throughout their lifecycle. 

This design introduces several design and implementation \textbf{challenges}. 
First, determining suitable memory regions to securely allocate these nonce buffers requires careful consideration to avoid leakage or unauthorized access. 
Managing the lifecycle of these nonces, such as reusing or invalidating them appropriately, must be tightly coupled with variable usage to prevent residual leakage or incorrect decryption. 
Second, establishing robust and efficient references between sensitive variables and their corresponding nonces is non-trivial. 
These references must support quick lookups and ensure integrity even under aggressive code transformations. 

\parh{Strategy 2: Secret-aware Register Allocation} refines Strategy 1 by aiming to protect secrets stored in sensitive stack areas through retention within registers throughout their lifetimes (see~\F~\ref{fig:variants}(b)). 
This approach seeks to eliminate sensitive memory writes, \textit{fundamentally mitigating ciphertext side-channel leakage} and thereby achieving stronger security guarantees.

A critical \textbf{challenge} lies in devising an effective algorithm for assigning a limited set of registers to tainted stack slots. 
Additionally, the algorithm must prioritize tainted stack slots for register assignment, while ensuring that the allocation does not excessively interfere with the original program’s register usage or degrade runtime efficiency. 
This involves analyzing variable lifetimes, usage frequency, and interference patterns to maximize register reuse without compromising the protection of secrets. 
Besides, because conventional compiler register allocation is primarily optimized for performance, it must be restructured to meet the new goal of preventing sensitive data from ever touching memory. 

\parh{Strategy 3: Diversion-based Obfuscation} enhances the performance of Strategy 2 by replacing masking operations for sensitive heap data. 
Instead, it disrupts observable ciphertext patterns by inserting decoy values into secret heap space, misleading the attacker and causing incorrect inferences about the actual secret values (see~\F~\ref{fig:variants}(c)).
To achieve this, Strategy 3 diversifies portions of the sensitive heap into a separately allocated buffer, populating the original memory layout with strategically placed decoy values.

This strategy introduces a critical \textbf{challenge}:
determining how to design diverse obfuscation strategies and how frequently these strategies should be applied is crucial. 
This requires a careful balance between performance and security: overly frequent obfuscation may incur unnecessary overhead, while insufficient obfuscation could leave exploitable patterns for attackers.
Moreover, the storage of secret values must be handled securely and efficiently, typically in dedicated buffers that are not reused by regular program data. 
These design considerations are crucial for maintaining the integrity of the obfuscation while minimizing performance penalties.

\section{Detailed Design of \tool}
\label{sec:design}

\subsection{Tainting Secret Locations}
\label{subsec:tainting}

Mitigating only truly vulnerable points would minimize the performance impact on protected programs. 
However, accurately detecting ciphertext side-channel leaks is difficult, with tools like \htool\ still generating false positives. Hence, we opt to taint all sensitive memory accesses. 
Besides, we adopt the same strategy as \ftool\ to log execution traces multiple times for each cryptography program with varied secrets and inputs. 
In addition to the direct usage of secrets, it is necessary to consider data derived from the secrets as ``sensitive'' as well.
A secret may appear in control-flow branches as a condition or in memory accesses as the index. Then variables guarded by a tainted condition are tainted as secret-related; variables assigned through memory access are tainted once the index of the buffer is tainted.

\tool\ provides extensive support for commonly used x86 instructions, covering data movement, arithmetic, logic operations, comparisons involving sensitive memory, updates to EFLAGS, zero/flag extensions, as well as frequently used SSE/AVX vector instructions.  
This range is sufficient to taint secret locations in most cryptographic applications. 

\subsection{Software-assisted Probabilistic Encryption}
\label{subsec:datamasking}

In software-assisted probabilistic encryption, the key is to introduce a random nonce into secret variables before they are written back into the memory, thereby achieving unpredictable ciphertexts. 

\begin{figure}[htbp]
    \centering
    \footnotesize
    \begin{tabular}{|ll|}
        \hline
        1:\ \ \textcolor{red}{load\ \ \ \ \ \ \ \ $MEM_{key}$\ \ \ \ $REG_{key}$} &//\ sensitive\ load\\
        2:\ \ load\ \ \ \ \ \ \ \ $MEM_{mask}$\ \ $REG_{mask}$ &//\ load\ mask\\
        3:\ \ save\ \ \ \ \ \ \ \ EFLAGS &\\
        4:\ \ xor\ \ \ \ \ \ \ \ \ $REG_{mask}$\ \ \ \ \textcolor{red}{$REG_{key}$} &\\
        5:\ \ restore\ \ \ \ \ EFLAGS &\\
        6:\ \ operate\ \ \ \ \textcolor{red}{$REG_{key}$} &\\
        7:\ \ save\ \ \ \ \ \ \ \ EFLAGS &\\
        8:\ \ update\ \ \ \ \ $REG_{mask}$ &//\ generate\ new\ mask\\
        9:\ \ xor\ \ \ \ \ \ \ \ \ $REG_{mask}$\ \ \ \ \textcolor{red}{$REG_{key}$} &\\
        10: restore\ \ \ \ \ EFLAGS &\\
        11: store\ \ \ \ \ \ \ $REG_{mask}$\ \ \ \ $MEM_{mask}$ &//\ store\ mask\\
        12: \textcolor{red}{store\ \ \ \ \ \ \ $REG_{key}$\ \ \ \ \ \ $MEM_{key}$} &//\ sensitive\ store \\
        \hline
    \end{tabular}
    \caption{In-place code insertion.}
    \label{fig:maskcode}
\end{figure}

\parh{Mitigation Code.}
\F~\ref{fig:maskcode} illustrates the general scenario of inserting mitigation code, where $MEM_{key}$ represents the memory cell of a sensitive memory access instruction.
The steps are: 1) loading the masked plaintext and corresponding random nonce (lines 1--2); 2) XORing the nonce with the masked plaintext to obtain true plaintext (lines 3--5); 3) performing calculations on the plaintext (line 6); 4) updating the random nonce and XORing it with the plaintext to obtain a new masked plaintext (lines 7--10); and 5) storing both the new masked plaintext and the new random nonce (lines 11--12).
It is worth noting that the \texttt{EFLAGS} register is protected using \texttt{lahf} and \texttt{sahf}, which temporarily save and restore the flags to ensure a consistent processor state during mitigation.

We determine the memory cells holding the random nonces for stack and heap areas with different methods.
In the stack, since the code is inserted during the register allocation phase in the LLVM backend, we can freely allocate a stack slot for the nonce. 
For example, a new stack slot, $MEM_{mask}$, is created to store the random nonce (line 11 of \F~\ref{fig:maskcode}) and is associated with the source memory $MEM_{key}$ (line 1). When loading the nonce, a load instruction is inserted to reference $MEM_{mask}$.
For heap memory, instead of intercepting memory allocation calls like \texttt{malloc}, \texttt{realloc}, \texttt{calloc}, and \texttt{free}, which would introduce significant overhead, we implement a more efficient hash-based mechanism. 
This scheme leverages the runtime heap address of the source memory $MEM_{key}$ to compute the index where the corresponding nonce is stored in the \texttt{.bss} section (see \S~\ref{subsec:buffermanage}). 

\parh{Random Nonce Generation.}
To enhance security, the random nonce is updated for each sensitive store instruction. 
In Strategy 1, a 1K random nonce buffer is pre-generated using \texttt{rdrand} during the cryptography program's initialization, stored in the \texttt{.bss} section. 
For unmasked stack areas with secrets, Variant \texttt{rdrand} selects a random nonce from this buffer as the initial nonce and increments it by three when storing new secrets at the same location. 
Other variants, instead, generate a random nonce in real-time using AES-NI or XorShift128+ schemes.
AES-NI requires a single instruction and two 128-bit registers such as \texttt{xmm14} and \texttt{xmm15}, while XorShift128+ needs 11 instructions and three 128-bit registers, from \texttt{xmm13} to \texttt{xmm15}.

\subsection{Secret-aware Register Allocation}
\label{subsec:registeralloc}

\parh{Feasibility Analysis.}
Building upon Strategy 1, the masking scheme for the stack area can be further enhanced through register allocation.
However, implementing this scheme is challenging due to limited register resources. To evaluate, we conducted a heuristic investigation. 
First, we identified the maximum number of stack slots involved in sensitive memory accesses among various cryptography programs from \T~\ref{tab:SSEimpact}, with libsodium's EdDSA implementation having the highest (583 slots). 
Accommodating this many vector registers in SIMD is difficult, prompting the exploration of a secret-aware register allocation approach. 
This involves tracking register liveness and timely deallocating frequently used stack slots in tainted functions. 
Next, we assessed the number of vector registers required. 
By analyzing disassembled cryptography programs, we found that the LLVM backend typically allocates the first 8 vector registers (\texttt{xmm0} - \texttt{xmm7}) for optimizing data movement, so we heuristically preserved the last 8 SSE registers (\texttt{xmm8} - \texttt{xmm15}) for sensitive stack slots. 
This approach resulted in an average performance impact of 7\%, ranging from 1\% in SHA512 (libsodium) to 20\% in AES (mbedTLS), indicating a viable and practical solution.

\begin{table}[htbp]
\centering
\caption{The maximum numbers of sensitive stack slots among tainted functions.}
\label{tab:SSEimpact}
\scriptsize
\begin{tabular}{ccccc}
\hline
    \multirow{2}{*}{\textbf{Implementation}}
    &\multirow{2}{*}{\textbf{Stack Slots}}
    &\multirow{2}{*}{}
    &\multicolumn{2}{c}{\textbf{Impact on performance}}\\
    \cline{4-5}
    & && \textbf{Cycles} & \textbf{Factor} \\
\hline
    libsodium-EdDSA &583 && 198331	&1.04 \\
    libsodium-SHA512 &27 && 43043	&1.01 \\
\hline
    mbedTLS-AES &7 && 571968	&1.20 \\
    mbedTLS-Base64 &25 && 11575	&1.03 \\
    mbedTLS-ChaCha20 &4 && 609942	&1.02 \\
    mbedTLS-ECDH &52 && 4586425	&1.08 \\
    mbedTLS-ECDSA &52 && 4095496	&1.04 \\
    mbedTLS-RSA &52 && 1943880	&1.03 \\
\hline
    OpenSSL-ECDH &157 && 1171024	&1.19 \\
    OpenSSL-ECDSA &79 && 18180467	&1.08 \\
\hline
    WolfSSL-AES &43 && 689316	&1.08 \\
    WolfSSL-ChaCha20 &5 && 512287	&1.06 \\
    WolfSSL-ECDH &100 && 362967	&1.04 \\
    WolfSSL-ECDSA &30 && 4559759	&1.08 \\
    WolfSSL-EdDSA &159 && 426239	&1.02 \\
    WolfSSL-RSA &28 && 543835	&1.11 \\
\hline
    Average &- && -	&1.07 \\
\hline
\end{tabular}
\end{table}

\parh{Register Allocation.}
When analyzing tainted functions in cryptography programs, we prioritize frequently accessed stack slots. 
These stack slots are more susceptible to ciphertext side-channel attacks due to their sequential overwrites. Moreover, preserving these slots in registers minimizes the overhead of masking operations. 
To manage this, we create two structures in the LLVM backend: \texttt{StackUsage}, which maps all sensitive stack slots to their respective MBB locations in the tainted functions, and \texttt{StackOpt}, which selects mappings from \texttt{StackUsage} based on the number of MBB locations for each stack slot until the available vector register capacity is reached. 
These structures enable efficient secret-aware register allocation, focusing on optimizing register use and reducing masking overhead for frequently accessed stack slots.

\begin{itemize}
\item \textit{Allocation:} When encountering a sensitive stack store where its slot is mapped in \texttt{StackOpt}, Strategy 2 assigns an available vector register to hold the variable. 
Simultaneously, the corresponding MBB location in \texttt{StackOpt} is removed to track the liveness of the allocated vector register.

\item \textit{Deallocation:} Once all MBB locations for a sensitive stack slot in \texttt{StackOpt} are removed, the associated vector register is freed, marking the end of its liveness.
Subsequently, Strategy 2 selects another stack slot from \texttt{StackUsage} based on the number of MBB locations and allocates available vector registers to the remaining stack slots in \texttt{StackOpt}.
\end{itemize}

\begin{lstlisting}[basicstyle=\scriptsize, frame=single, caption=Sensitive stack slots contained in MBBs., label=reglist, escapeinside=``]
42: entry, if.end19, if.end24, if.end30...
38: entry, if.end24, if.end30...
39: entry, while.body, if.end30...
52: if.end30...
49: if.end30...
46: while.body, if.then10, if.end12, if.then18,
    if.end19...
43: while.body, if.end19, if.then22, if.then28,
    if.end30...
40: entry, while.cond...
50: if.end30...
44: while.body, if.end19, if.end24, if.end30...
\end{lstlisting}

\begin{figure}[htbp]
\centering
\includegraphics[width=0.95\linewidth]{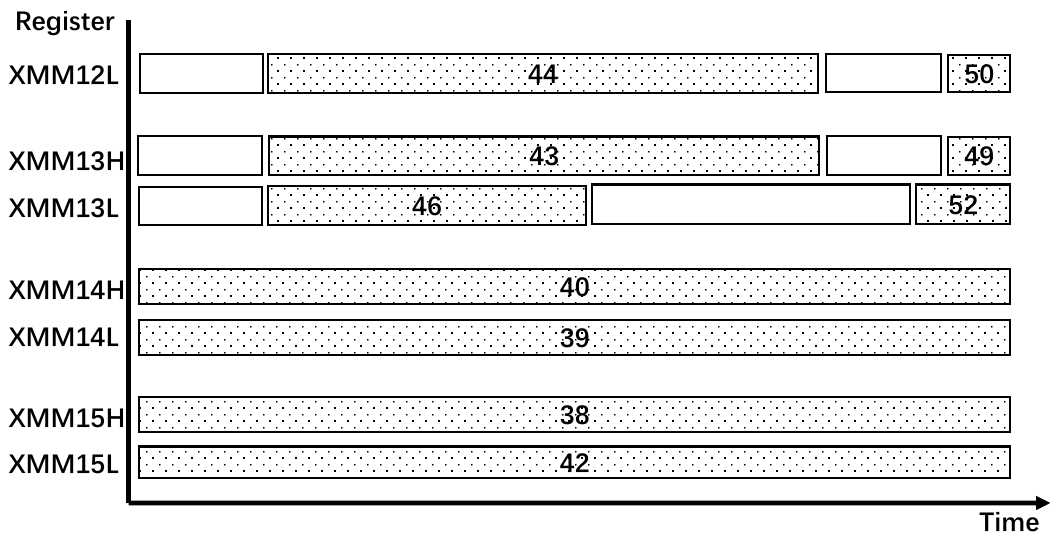}
\caption{An example of register allocation from the function \texttt{bn\_mul\_add\_words} of OpenSSL-ECDSA. The white and shaded blocks represent the liveness of stacks, with shaded blocks containing numbers that denote registers holding the sensitive stack slots.}
\label{fig:liveness}
\end{figure}

We illustrate the register allocation process for sensitive stack slots in the function \texttt{bn\_mul\_add\_words} from OpenSSL-ECDSA. Listing~\ref{reglist} provides a sorted \texttt{StackOpt} structure, showing sensitive stack slots and their positions in MBBs.
The process is simulated up to the \texttt{if.end30} MBB and subsequent MBBs are omitted for brevity. 
As shown in \F~\ref{fig:liveness}, 16 stack slots, each 8 bytes in size, are allocated across vector registers \texttt{xmm8} to \texttt{xmm15}, with \texttt{H} and \texttt{L} representing the high and low 64 bits of each register. 
Registers such as \texttt{xmm13L}, \texttt{xmm13H}, and \texttt{xmm12L} are recycled and reallocated to stack slots 52, 49, and 50, respectively, demonstrating efficient register usage through liveness tracking.

\subsection{Diversion-based Obfuscation}
\label{subsec:obfuscation}

Based on Strategy 2, we further improve heap mitigation efficiency by introducing a diversion-based obfuscation approach. 
This strategy is inspired by cache side-channel defenses, which execute both decoy and true paths while committing only the true results~\cite{crane2015thwarting,rane2015raccoon}. 
Overall, instead of applying masking to every sensitive heap unit, we obfuscate secret values by inserting a decoy value (varying nonce) into the original heap position and diverting the true secret into a secure buffer located in the \texttt{.bss} section or vice versa.
An easy and effective way to index the true secret from the secure buffer is to follow the indexing scheme used for nonce management for heap space in Strategy 1.

In this strategy, the attacker is assumed to have enhanced capabilities, including the ability to observe changes in the ciphertext of the secure buffer located in the \texttt{.bss} section.
Therefore, to defend against such an attacker, Strategy 3 ensures an obfuscation process that deliberately confuses and misleads the attacker’s observations, preventing accurate reasoning based on ciphertext changes:

\begin{itemize}
\item \textbf{Obfuscation 1}: Ciphertext changes at both the original secret location (inserting a decoy value, like a varying nonce) and the secure buffer (diverting the true secret).

\item \textbf{Obfuscation 2}: Ciphertext Changes at only one location are deliberately designed to confuse the attacker, making it unclear which position actually holds the true secret.
\end{itemize}

\begin{figure}[htbp]
\centering
\includegraphics[width=0.95\linewidth]{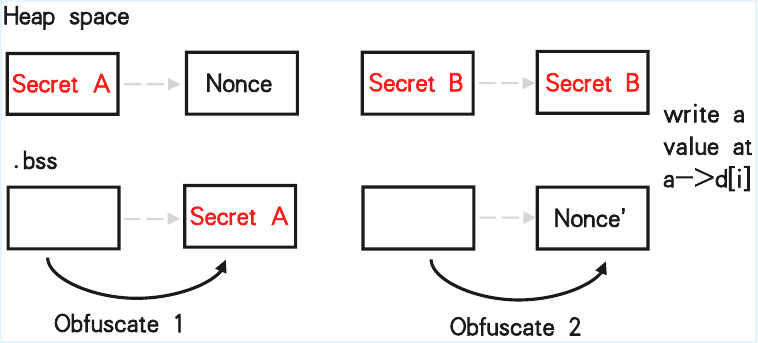}
\caption{The observation from an attacker when performing multiple obfuscation at the same a->d[i].}
\label{fig:obfuscation}
\end{figure}

An example where ``Secret A" and ``Secret B" are stored in heap space and undergoes constant-time swap, similar to the process shown in \F~\ref{fig:channel}(b), is used to demonstrate the obfuscation process described above, as illustrated in \F~\ref{fig:obfuscation}.
An instance of \textbf{Obfuscation 1} occurs on ``Secret A" when a distinct nonce is inserted into the original heap location, causing the ciphertext at the original secret position to change accordingly.
At the same time, ``Secret A" is moved into the secure buffer, resulting both ciphertexts change at original secret position and the secret buffer.
For the ``Secret B", an instance of \textbf{Obfuscation 2} occurs by only updating the nonce to the secret buffer.
From the attacker's perspective, it remains unclear whether the true secret resides in the unchanged ciphertext within the original heap location or the modified ciphertext at the secret buffer.

The implementation of Strategy 3 pursues obfuscating the entire sensitive heap space by performing the two obfuscation strategies evenly.
By deterministically selecting a subset, such as using even-odd indexing, Strategy 3 achieves low performance overhead while still fulfilling the requirements of the obfuscation mechanism.
The insertion of varying decoy values follows the approach of Variant \texttt{rdrand} from Strategy 1, where simply incrementing the nonce each time is sufficient to produce a distinguishable ciphertext.

\subsection{Managing Nonce Buffers}
\label{subsec:buffermanage}

In \tool, all masking-related mitigation manages two types of nonces: the \textit{initial nonce} and the \textit{currently used nonce}.

\parh{Random Nonce Buffer.}
As introduced in \S~\ref{subsec:datamasking}, \tool pre-generates a buffer in the \texttt{.bss} section with 1K random nonces using \texttt{rdrand} during initialization. 
For unmasked stack or heap areas, \tool\ calculates an index based on the memory address (\textit{addr}) to select the corresponding random nonce from the buffer. 
The index is computed as $index = addr\ \&\ 0x3FF$, and the random nonce is retrieved from the \texttt{.bss} using $randomArray(, index, 8)$, where $randomArray$ is the starting address of the buffer.

\parh{Currently Used Nonces.}
Two tactics are adopted to store random nonces for sensitive data in the stack and heap areas, enabling their subsequent decoding.
For the stack, the compiler automatically allocates slots for nonces and associates them with the corresponding memory, accessed by applying an offset to the \texttt{rsp} register.
This method ensures the nonces are isolated from other stack locations by using the \texttt{rbp} register. 
Additionally, the nonce buffer, used to store active nonces, is initialized to zero and set in the entry block of each function.
Unlike \ftool, which risks overlap by storing nonces outside the runtime stack area, our approach avoids potential malfunctions.

For the heap, instead of applying a constant offset for the nonce buffer like in \ftool, we place the nonce buffer in the \texttt{.bss} section as well, pre-allocating sufficient space in advance. 
Mapping 64-bit addresses to the random nonce buffer requires a balance between avoiding sparsity and ensuring no collisions. 
To achieve this, we use Fipolach hashing. Since the last 20 bits of the address vary, we compute the multiplier as $2^{20}$$\times$\textit{golden-ratio}, approximately 648,056, and carefully select a shift value of 22 to prevent collisions. 
Therefore, a heuristic hash function maps heap addresses to buffer indices for retrieving nonces: $((addr\ \&\ 0xFFFFF) * 648056) \gg 22$.
In this setting, 128K consecutive 8-byte entries are mapped independently into the 1MB nonce buffer without collisions since all variants of \tool\ align masking operations to 8-byte addresses. 
However, the maximum index generated by the hash function is 162,012, which exceeds 128K. To accommodate this, a 256K-entry buffer is required, resulting in a 2MB nonce buffer with around 50\% utilization in our configuration.
Through experimentation, we found that a nonce buffer space of 2MB is sufficient for cryptographic programs without collisions.


To handle potential hash collisions in nonce buffers for heap areas, such as in server environments, we propose dynamically expanding the buffer by allocating multiple groups of entries for each index (10 groups per index, requiring a 20MB nonce buffer). 
If two different heap addresses generate the same index, the first 8 bytes are used to match the heap address, allowing us to locate the corresponding nonce in the subsequent 8 bytes. 
We evaluated this method by simulating collisions across all 10$\times$128K entries. The results showed only a slight increase of time to initialize the expanded nonce buffer and handle each collision ($\approx 8 ms$), indicating minimal performance impact.

\subsection{Correctness of Transformation}
\label{design:correctness}

We refine the semantic equivalence proof by making the simulation argument explicit for each mitigation macro.

\parh{Operational Semantics.}
Let $\mathit{P}$ and $\mathit{P}'$ denote the state spaces of the original and transformed programs, respectively.
A state $\mathit{P}'$ in the transformed semantics is defined to include temporary variables introduced by mitigation steps:
\[
  \sigma' = (pc', R', M', A', \Delta),
\] 
where $pc$ is the program counter, $R : \mathit{Reg} \to \mathit{Word}$ is the register file, $M : \mathit{Addr} \to \mathit{Word}$ is the memory store, $A$ denotes auxiliary runtime structures (e.g., call stack, heap metadata), and $\Delta$ is a finite auxiliary environment recording temporary registers, nonce updates in progress, and dummy writes that have not yet been logically committed.

We use the small-step relation $\sigma \xrightarrow{e} \sigma'$ to denote one micro-step of the transformed execution, possibly corresponding to a single original instruction or part of a macro expansion.

\parh{Abstraction Function.}
We define the abstraction function $\alpha: \mathit{P}' \to \mathit{P}$ to erase all mitigation artifacts:
\[
  \alpha(pc', R', M', A', \Delta) = (pc', \widehat{R}, \widehat{M}, \widehat{A}),
\]
where $\widehat{R}(r) = R'(r) \oplus \mathit{Nonce}'_o$ if $r$ holds masked value of object $o$, and $\widehat{M}(a) = M'(a) \oplus \mathit{Nonce}'_o$ if $a$ holds masked value for $o$.
All entries in $\Delta$ (nonce temporary buffers, dummy buffer slots, and volatile intermediates) are discarded.
This abstraction ensures that every mitigated state corresponds to a well-formed original state, preserving all functional contents.

\begin{lemma}[Atomicity of Encode–Update]
Let $\sigma \xrightarrow{*} \sigma'_1$ represent a finite micro-trace corresponding to one logical masked store (masking plaintext and updating nonce). If the macro is correctly emitted (single entry, single exit, no early branch), then $\alpha(\sigma) \xrightarrow{store(v)} \alpha(\sigma'_1)$ in the semantics of $\mathit{P}$. In particular, nonce update and ciphertext write are atomic with respect to the abstract semantics:
\[
\widehat{M}'_1(a) = v, and\ \alpha(\mathit{Nonce}'_1) = f(\alpha(\mathit{{Nonce}})),
\]
where $f$ represents the nonce update function, so that decoding after the step yields the same logical value as $\mathit{P}$.
\end{lemma}

\begin{proof}
Each masked store in $\mathit{P}'$ expands to a constant sequence of side-effect–free micro-steps that (i) update $\mathit{Nonce}$ by $f(\mathit{Nonce})$, (ii) compute $v \oplus \mathit{Nonce}'_1$ into a temporary register, and (iii) write to the memory location of this store.
Because the sequence has no secret-dependent control flow and no shared aliasing on intermediate locations, we can compose it as one atomic transition in $\alpha$’s image.
\end{proof}

\begin{lemma}[Register Integrity]
Let $r$ be a register holding a secret-typed SSA value.
In $\mathit{P}'$, all $r \in \mathit{Reg}_{sec}$ are allocated from a non-spillable class (xmm8–xmm15).
For any call or control transfer step $\sigma'_i \xrightarrow{e} \sigma'_{i+1}$, where $e$ does not involve secret declassification, 
\[
\widehat{R}'_{i+1}(r) = R'_i(r),
\]
and no spill of secret value occurs in memory between $\sigma'_i$ and $\sigma'_{i+1}$.
\end{lemma}

\begin{proof}
Register allocation emits static proof annotations guaranteeing that all secret-typed temporaries are mapped to the reserved register set and excluded from callee-saved or scratch registers.
Since the register class is non-spillable and all save/restore code paths are type-checked, the corresponding abstract state is invariant under call transitions.
\end{proof}

\begin{lemma}[Dummy Erasure]
Let $M'$ be the memory after an obfuscation–diversion macro expansion.
Then for any logical address $a$ in the program memory,
\[
\widehat{M}'(a) = M(a)\ if\ a\ is\ a\ true\ target, and\ \widehat{M}'(a)\ is\ ignored\ otherwise.
\]
Hence dummy writes to the original secret location or secure buffers are erased under $\alpha$ and have no functional effect.
\end{lemma}

\begin{proof}
The diversion macro writes to the true address $a_{real}$ followed by a corresponding dummy address $a_{dummy}$.
Since $a_{dummy}$ is discarded by $\alpha$, the dummy step is semantically inert.
\end{proof}

\begin{theorem}[Trace Equivalence]
Let $\xrightarrow{*}$ denote the reflexive-transitive closure of the step relation.
For any initial states $\sigma_0 \in \mathit{P}$ and $\sigma_0' \in \mathit{P'}$ satisfying $\alpha(\sigma_0')=\sigma_0$,
if
\[
\sigma_0 \xrightarrow{e_1} \sigma_1 \xrightarrow{e_2} \cdots \xrightarrow{e_n} \sigma_n
\]
is a valid trace in $\mathit{P}$, then there exists a corresponding micro-trace in $\mathit{P}'$:
\[
\sigma'_0 \xrightarrow{*} \sigma'_1 \xrightarrow{*} \cdots \xrightarrow{*} \sigma'_n
\]
such that $\alpha(\sigma_i')=\sigma_i$ and the sequence of functional observables $(a_1,\dots,a_n)$ is preserved.
Conversely, every terminating trace of $\mathit{P'}$ maps under $\alpha$ to a valid trace of $\mathit{P}$.
\end{theorem}

\begin{proof}
By structural induction on the length of the trace and case analysis on instruction type.
For each case, apply Lemmas 4.1–4.3 to show that the composite micro-steps of $\mathit{P}'$ correspond to a single abstract step of $\mathit{P}$.
Since $\alpha$ erases all mitigation artifacts (dummy buffers, temporaries, nonce intermediates) and preserves all functional writes/reads, the projections of both traces are equivalent.
Therefore, $\mathit{P}'$ is semantically equivalent to $\mathit{P}$ with respect to all functional observables.
\end{proof}

\parh{End-to-End Correctness.} Under the compiler’s emission discipline and static verification of the obligations,
\[
\mathit{P}' \approx \mathit{P}\ w.r.t.\ functional\ semantices.
\]
Hence, the transformation preserves correctness.

\section{Security Analysis}
\label{sec:security}

\subsection{Coverage Model and Theoretical Guarantees}
\label{subsec:coverageproof}

Both \tool\ and binary instrumentation baselines rely on dynamic taint analysis to identify secret-origin computations. Dynamic tainting, by construction, captures \textit{only} those control-flow paths exercised under the observed inputs.
Let $\mathcal{T}$ denote the set of instructions observed in all dynamic traces under a finite input set $X$.
A repair mechanism based purely on dynamic tainting ensures protection for $\mathcal{T}$ but provides no guarantee for instructions outside $\mathcal{T}$ that may be executed under unseen inputs.
Thus, the theoretical coverage of pure dynamic repair is bounded by the union of observed traces:
\[
\mathsf{Cov}_{dyn}(X) = \bigcup_{x\ \in\ X}Trace(P, x).
\]

\tool\ extends this model by coupling tainted results with compiler-based static propagation.
Let $\mathsf{SSA}(P)$ denote the program’s single static assignment graph, where edges represent def-use dependencies resolved at compile time.
When a variable $v$ is marked as secret-tainted in any dynamic trace, the compiler statically propagates this taint along $\mathsf{SSA}(P)$ to all syntactically dependent uses of $v$, including those that appear in unexecuted branches.
The resulting repair set is:
\[
\mathsf{Cov}_{cg}(X) = Closure_{ssa}(\mathsf{Cov}_{dyn}(X)),
\]
which over-approximates the dynamic coverage.
Hence, while \tool\ does not assume complete or sound taint propagation in the theoretical sense, it guarantees that \textit{for any secret-origin value observed dynamically}, all of its syntactic dependents across the program’s static representation are repaired.
In contrast, binary-level repair operates post-compilation and cannot apply equivalent propagation because its repair set is limited to the executed instruction subset $\mathsf{Cov}_{dyn}(X)$, leaving unexecuted branches unpatched.

This theoretical model explains why \tool\ can claim \textit{path-extended coverage}: every execution path in $\mathit{P}'$ that includes a use of any previously observed secret is guaranteed to have all of its relevant stores and spills repaired prior to execution.
The baseline, lacking such static closure, remains \textit{path-local}: its guarantees hold only for the dynamic traces it observed.

\subsection{Security of Strategies}
\label{subsec:strategyproof}

\parh{Baseline characteristics and limitations.}
We consider a baseline that performs single masking at the binary level: on tainted stores it writes $m \oplus v$ using freshly generated $m$ (or otherwise safe nonce management); loads decode analogously.
Both the baseline and \tool\ rely on \textit{dynamic taint analysis} to identify sensitive memory writes, so the coverage of \textit{executed, explicit} memory access instructions within a run is comparable. 
However, baseline repair is \textit{path-local}: instrumentation is emitted only for the tainted instructions that appear along the observed dynamic path and cannot \textit{statically} propagate repair to unexecuted branches before they execute. 
Moreover, the baseline operates post-compilation and thus (i) cannot eliminate secret-dependent \textit{timing/access pattern} side channels induced by the original code layout and memory effects, (ii) cannot retroactively prevent spills of ciphertext values, since register allocation decisions have already been fixed, and (iii) cannot relocate or diversify the storage of protected ciphertexts, as all masked values must still be written to their original memory locations. 
We quantify an adversary's advantage by ${Adv}(\mathcal{A}\mid \cdot)$.

\parh{Strategies S1--S3: formal obligations and lemmas.}
In \textit{S1 (Masking with nonce buffers)}, each sensitive object $o$ is deterministically mapped to a dedicated nonce slot; on the $k$-th store of $v_k$ to $o$, transformed program $P'$ updates $Nonce_{cur} \gets f(Nonce_{cur})$ with fixed, secret-independent $f$ and writes $v_k \oplus Nonce_{cur}$; reads decode using the current nonce $Nonce_{cur}$.
\textbf{Obligations:} (O1) every statically identifiable access to $o$ in $P'$ is wrapped by the encode/decode scheme; (O2) distinct objects map to disjoint nonce slots; (O3) nonce access/update code is constant-time and secret-independent; (O4) each logical store emits exactly one encode and one update (atomicity).
\begin{lemma}[Ciphertext-channel bound under S1]
Under (O1)--(O4), for any run and any two secrets $s_0,s_1$, decoding in $P'$ yields the same logical values as in $P$. 
Moreover, unlike the baseline which provides only \textit{path-local repair}, S1 achieves \textit{path-extended coverage}: all syntactic accesses to $o$ in the compiled program, including those residing in unexecuted branches, are instrumented ahead of execution. 
Consequently, ciphertext side-channel resistance holds across the entire program state space, ensuring stronger security guarantees than binary-level masking.
\end{lemma}

For \textit{S2 (Security-aware register allocation)}, secret values are assigned to a non-spillable register class; if constraints are not met, the secret is masked before any store (fail-closed).
\textbf{Obligations:} (R1) a post-allocation audit certifies that no store in $P'$ has a secret-tainted source unless it passes through an explicit masking; (R2) it preserves the invariant and introduces masking when architectural stores are unavoidable.
\begin{lemma}[Ciphertext-channel elimination under S2]
Under (R1)--(R2), for register-resident secrets the memory-ciphertext channel is \textit{absent} in $P'$, i.e., along all paths there is no architectural store whose source depends on a secret except at audited, masked boundaries. 
This strictly dominates the baseline, which operates post-compilation and thus cannot retroactively prevent secret-dependent spills: the baseline can only repair the resulting store after the fact, leaving the ciphertext content present (though masked).
\end{lemma}

With \textit{S3 (Diversion-based Obfuscation)}, each sensitive store is compiled into a fixed macro: first the real store to the intended address, then a dummy store to the remaining address, with order, sizes, alignment, and fences fixed; dummy selection/rotation is secret-independent.
\textbf{Obligations:} (D1) every sensitive store site in $P'$ is wrapped by the macro; (D2) the macro has no secret-dependent control flow or access pattern; (D3) dummy buffers are included in program-visible memory.
\begin{lemma}[Side-channel non-interference under S3]
Under (D1)--(D3), the access patterns and timing differences are indistinguishable for fixed inputs and varying secrets.
This is strictly stronger than the baseline, which inherits the original program's secret-dependent access pattern/timing behavior and cannot eliminate these channels post hoc.
\end{lemma}

\parh{Composition theorem.}
Let $D^\star = S3 \circ S2 \circ S1$ and assume all obligations (O1)--(O4), (R1)--(R2), (D1)--(D3) hold via compiler-side checks.
Then, for the protected region $R$ and any two secrets $s_0,s_1$, (i) by S3, access pattern or timing side channels are inexistent; (ii) by S2, for register-resident secrets the ciphertext side channel is eliminated, since spills are prevented at the register allocation stage; (iii) by S1, for the remaining necessary stores the ciphertext side channel leakage is less than path-local repair.
Consequently,
\[
  {Adv}(\mathcal{A}\mid D^\star) \;<\; {Adv}(\mathcal{A}\mid \text{baseline}),
\]
since the baseline cannot (i) pre-instrument unexecuted branches, (ii) eliminate access pattern/timing channels, nor (iii) prevent spills of secret values before they occur (it can only repair their resulting stores). 
Importantly, this comparison does not rely on nonce reuse weaknesses: both approaches are assumed nonce-safe.
Hence $D^\star$ yields strictly stronger (channel elimination and non-interference) and deeper (a priori path coverage and verifiable audits) security guarantees than baseline masking.

\subsection{Security of Implementation}
\label{subsec:toolproof}

\parh{Locations of Nonce Buffers.}
In \tool, the nonce buffers contain both pre-generated random numbers and currently used nonces, stored in the \texttt{.bss} section. 
The security of masking operations across all strategies relies on the integrity and confidentiality of these nonce buffers to prevent leakage. 
While accessing these buffers may create execution traces that could be exploited, their security is protected by Address Space Layout Randomization (ASLR), which loads the \texttt{.bss} section into different memory locations each time. 
Attackers would need to exploit memory vulnerabilities in the software, which is not the main focus of \tool's protections. 
Thus, \tool\ provides similar security guarantees as \ftool\ for nonce buffers.
Furthermore, even if ASLR is bypassed~\cite{jang2016breaking}, the security of the nonce buffers remains robust due to SEV memory encryption, which ensures that these buffers are always encrypted.
The pre-generated nonce buffer stays consistent, leading to no variation in its ciphertext, while the currently used nonce is updated with each use during masking operations, resulting in unpredictable ciphertext.

A similar security analysis is applied to the secure buffer in Strategy 3, ensuring that an attacker cannot extract valid information from the changing decoy values (varying nonces) or the swapped secrets.

\parh{Re-use of Nonces.}
In \tool, the nonce selection is based on the address of the target secret, leading to nonce reuse for secrets at the same address. When a function repeatedly calls the same callee, the stack frames share the same address space. 
However, \tool\ maintains security through two key mechanisms: varying parameters that keep ciphertext unpredictable under SEV and a masking process that safeguards the callee against information leakage. 
As a result, an attacker can only deduce that the same secret is being used by noticing identical ciphertext sequences.

\begin{figure}[htbp]
\centering
\includegraphics[width=0.95\linewidth]{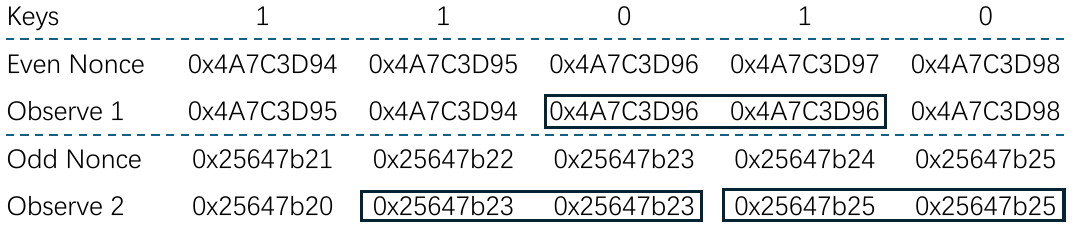}
\caption{A corner case arises when one-bit secret changes.} 
\label{fig:increment}
\end{figure}

\begin{figure*}[ht]
\centering
\includegraphics[width=0.9\textwidth]{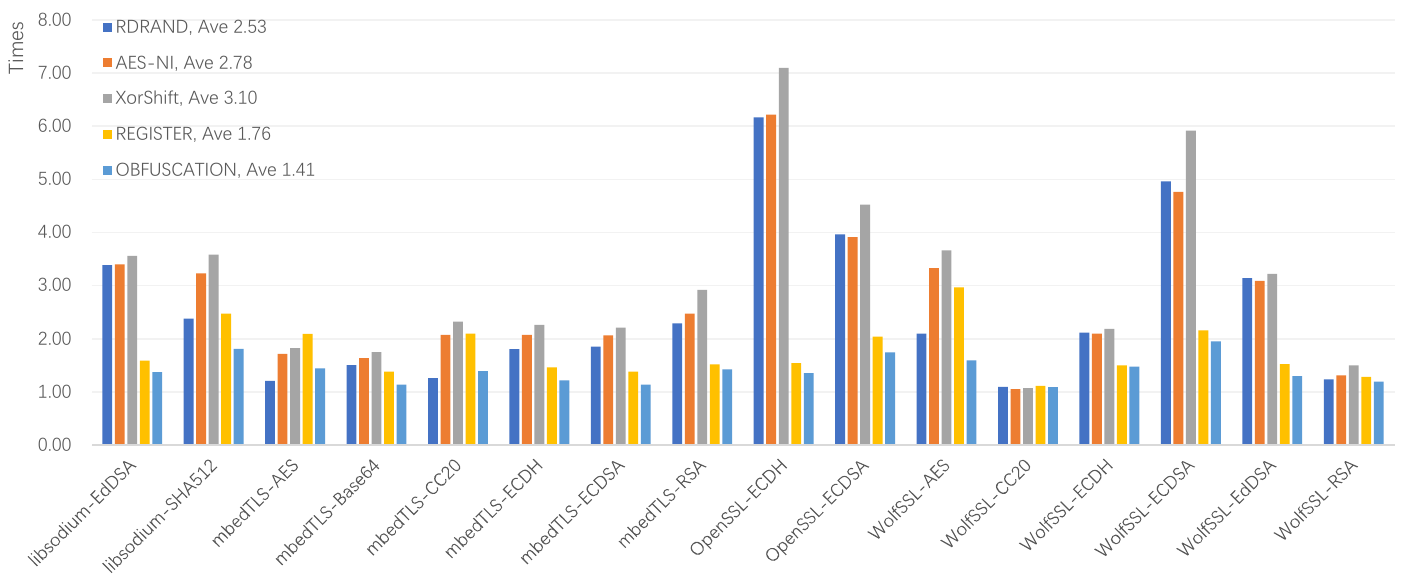}
\caption{Performance statistics towards mitigated cryptography software with 3 strategies of \tool. XS+ is short for XorShift128+. CC20 is short for ChaCha20.}
\label{fig:performance}
\end{figure*}

\begin{table*}[htbp]
\centering
\caption{Detailed performance statistics of \tool. XS+ is short for XorShift128+. CC20 is short for ChaCha20.}
\label{tab:resultsoverview}
\scriptsize
\resizebox{\linewidth}{!}{
\begin{tabular}{c|cccccc|cc|cc}
\hline
    \multirow{2}{*}{\textbf{Implementation-Orig}}
    &\multicolumn{6}{c|}{\textbf{Variant 1}}
    &\multicolumn{2}{c|}{\textbf{Variant 2}}
    &\multicolumn{2}{c}{\textbf{Variant 3}}\\
    & \textbf{RDRAND} & \textbf{Factor}
    & \textbf{AES} & \textbf{Factor} & \textbf{XS+} & \textbf{Factor}
    & \textbf{REGISTER} & \textbf{Factor}
    & \textbf{OBFUSCATION} & \textbf{Factor}\\
\hline
    libsodium-EdDSA-191618	&648861	&3.39 &	651930	&3.40 	&682159	&3.56 &	304067	&1.59 & 262942 & 1.37\\
    libsodium-SHA512-42482	&101088	&2.38 &	137298	&3.23 	&152227	&3.58 &	105061	&2.47 & 76812 & 1.81\\
\hline
    mbedTLS-AES-474926	&574695	&1.21 &	815542	&1.72 	&866192	&1.82 &	993255	&2.09 & 684987 & 1.44\\
    mbedTLS-Base64-11220	&16877	&1.50 &	18360	&1.64 	&19642	&1.75 &	15493	&1.38 & 12769 & 1.14\\
    mbedTLS-CC20-597000	&752978	&1.26 &	1236573	&2.07 	&1385737	&2.32 &	1251150	&2.10 & 831700 & 1.39\\
    mbedTLS-ECDH-4253887	&7680051	&1.81 &	8819828	&2.07 	&9626046	&2.26 &	6211978	&1.46 & 5189010 & 1.22\\
    mbedTLS-ECDSA-3942140	&7310694	&1.85 &	8133126	&2.06 	&8717678	&2.21 &	5448579	&1.38 & 4489770 & 1.14\\
    mbedTLS-RSA-1884568	&4314759	&2.29 &	4658550	&2.47 	&5504460	&2.92 &	2861580	&1.52 & 2681874 & 1.42\\
\hline
    OpenSSL-ECDH-986228	&6081644	&6.17 &	6132829	&6.22 	&6997656	&7.10 &	1519590	&1.54 & 1334321 & 1.35\\
    OpenSSL-ECDSA-16854915	&66816768	&3.96 &	65915307	&3.91 	&76290783	&4.53 &	34374077	&2.04 & 29399735 & 1.74\\
\hline
    WolfSSL-AES-639055	&1341064	&2.10 &	2126681	&3.33 	&2340311	&3.66 &	1897460	&2.97 & 1018463 & 1.59\\
    WolfSSL-CC20-481650	&526643	&1.09 &	507535	&1.05 	&517214	&1.07 &	537044	&1.12 & 525956 & 1.09\\
    WolfSSL-ECDH-348690	&737935	&2.12 &	731250	&2.10 	&761804	&2.18 &	522330	&1.50 & 515347 & 1.48\\
    WolfSSL-ECDSA-4226940	&20976600	&4.96 &	20141736	&4.77 	&25013569	&5.92 &	9132870	&2.16 & 8254290 & 1.95\\
    WolfSSL-EdDSA-419507	&1317104	&3.14 &	1296690	&3.09 	&1350810	&3.22 &	639840	&1.53 & 543600 & 1.30\\
    WolfSSL-RSA-491970	&608190	&1.24 &	645980	&1.31 	&737520	&1.50 &	631806	&1.28 & 586290 & 1.19\\
\hline
    Average	&-	&2.53 &	-	&2.78 	&-	&3.10 &	-	&1.76 & - & 1.41\\
\hline
\end{tabular}}
\end{table*}

\parh{Increment-by-Three in Nonces.}
Updating nonces in Variant \texttt{rdrand} is essential for ensuring \tool's security. 
When consecutive keys differ significantly, masked secrets produce unpredictable ciphertexts. 
However, in cases where the key changes by only one bit, a weak nonce update (e.g., incrementing by one) can expose this change through identical masked plaintexts. 
For example, as illustrated in \F~\ref{fig:increment}, observing masked plaintexts starting from even and odd nonces can leak up to four key bits, weakening security.
This weak nonce update scheme can be resolved by incrementing the nonce by three, ensuring that at least two bits of the masked plaintext differ between updates. 
This approach blocks identical ciphertexts from revealing bit changes without performance penalty.

\section{Evaluation}
\label{sec:evaluation}

\subsection{Implementation and Experiment Setup}
\label{sec:implementation}

\tool\ is built within the LLVM-9 framework, consisting of 4.6K lines of C++ code, but \textbf{it is not dependent on any LLVM-specific features}, making it portable to other compiler frameworks (discussed further in \S~\ref{sec:discussion}). 
The source code is publicly available online, with ongoing plans to migrate it to a more recent version of LLVM.
Developers can easily use \tool\ by manually identifying secrets in cryptography programs and labeling them with the DFSan API, such as inserting \texttt{dfsan\_set\_label(label, key, sizeof(key))} record both the pointer and size of the secret.
Afterward, the taint analysis and mitigation processes are handled automatically without further intervention.

We evaluated \tool\ with several cryptography libraries, i.e., libsodium-1.0.18, mbedTLS-3.3.0, OpenSSL-3.0.2, and WolfSSL-5.3.0. 
All experiments were performed on an AMD EPYC 7513 CPU with SEV-SNP enabled, running Ubuntu 20.04. 
For comparison with \ftool, the same library snapshots used in \ftool\ evaluations were tested. 
Default optimization levels were applied, with \texttt{-O3} for OpenSSL and \texttt{-O2} for the other libraries.

\subsection{Performance Overhead}
\label{subsec:overhead}

We use \texttt{rdtsc} to accurately measure the execution cycles. 
The performance overhead (averaged over 200 iterations) for the patched libraries using \tool\ are illustrated in \F~\ref{fig:performance} (see \T~\ref{tab:resultsoverview} for detailed statistics).

Variant \texttt{rdrand} in Strategy 1 incurs an average overhead of 2.53$\times$, ranging from 1.09$\times$ in the patched ChaCha20 of WolfSSL to 6.17$\times$ in OpenSSL's ECDH. 
Variant AES-NI and XorShift128+ use two methods for on-the-fly mask generation, resulting in average overheads of 2.78$\times$ and 3.10$\times$. 
These approaches slightly reduce performance compared to Variant \texttt{rdrand}'s optimized scheme. 
The higher overhead in XorShift128+ is due to its multiple-instruction process, which takes longer than the single \texttt{vaesenc} instruction from AES-NI.
Strategy 2 improves performance by keeping sensitive data in registers, achieving a lower average overhead of 1.76$\times$ compared to Variant \texttt{rdrand}'s 2.53$\times$. 
However, this performance gain is not consistent across all cryptography libraries due to the trade-off in SSE register usage. While fewer mask instructions are needed, the reduced availability of SSE registers can impact other parts of the program.
Strategy 3 delivers the best performance, incurring a geo-mean overhead of only 1.41$\times$, with a maximum of 1.95$\times$.  
This improvement is achieved by selectively obfuscating memory units, rather than applying masking to every sensitive memory access, as done in Strategy 1.

\subsection{Factors Contributing to Overhead}
\label{subsec:factors}

\parh{\texttt{rdrand} as the Bottleneck.}
Employing on-the-fly random number generation as nonces offers strong security but can lead to significant performance overhead, especially when relying on \texttt{rdrand}, as seen in \ftool-\textsc{Fast}, which incurs an average overhead of 16.8$\times$ and a peak of 40$\times$. 

To evaluate the performance benefits of Variant \texttt{rdrand} of \tool\ compared to the on-the-fly method, we examined both asymmetric cryptography libraries (RSA and ECDSA) and symmetric cryptography libraries (AES and ChaCha20). 
The results in \T~\ref{tab:variant1benefit} show significant improvement for asymmetric libraries, averaging 6.79$\times$, while the improvement for symmetric libraries is more modest at 1.39$\times$. 
This demonstrates the effectiveness of pre-generating a nonce buffer in enhancing performance, especially for asymmetric cryptography.

\begin{table}[htbp]
\centering
\caption{Performance improvement by pre-generated nonce buffer. CC20 is short for ChaCha20.}
\label{tab:variant1benefit}
\scriptsize
\begin{tabular}{cccc}
\hline
    \multirow{2}{*}{\textbf{Implementation}}
    &\multicolumn{1}{c}{\textbf{Strategy 1}}
    &\multicolumn{1}{c}{\textbf{On-the-fly \texttt{rdrand}}}
    &\multirow{2}{*}{\textbf{Promotion}}\\
    & \textbf{Factor} & \textbf{Factor} & \\
\hline
    ECDSA-mbedTLS & 1.85 & 10.87 & 5.87 \\
    ECDSA-OpenSSL & 3.96 & 18.92 & 4.77 \\
    ECDSA-WolfSSL & 4.96 & 18.60 & 3.75 \\
    RSA-mbedTLS & 2.29 & 14.91 & 6.51 \\
    RSA-WolfSSL & 1.24 & 16.21 & 13.07 \\
\hline
    AES-mbedTLS & 1.21 & 2.38 & 1.96 \\
    AES-WolfSSL & 2.10 & 3.10 & 1.47 \\
    CC20-mbedTLS & 1.26 & 1.30 & 1.03 \\
    CC20-WolfSSL & 1.09 & 1.24 & 1.13 \\
\hline
    Average & 2.21 & 9.72 & 4.39 \\
\hline
\end{tabular}
\end{table}

\parh{Alternative Nonce.}
Unlike Variant \texttt{rdrand}, which selects and increments a random number for each sensitive memory write, Variant AES-NI and XorShift128+ generate a random nonce directly when an updated nonce is needed. 
AES-NI achieves this with a single \texttt{vaesenc} instruction, while XorShift128+ requires multiple instructions, resulting in slightly higher overhead. 
Additionally, these two variants require reserving a portion of the SSE vector registers (two for AES-NI and three for XorShift128+), which can affect other parts of the cryptography library, adding more overhead compared to Variant \texttt{rdrand}, as reflected in performance results.

\parh{The Profit of Registers.}
The primary benefit of using registers to store intermediate results is the reduced memory access overhead. 
However, the performance advantage of Strategy 2 over Variant \texttt{rdrand} is inconsistent across all cryptographic applications. 
We hypothesize that Strategy 2's reliance on SSE registers may limit their availability in other program areas. 
To test this, we introduce a new strategy called ``RSV-SSE'', which combines the approach of Variant \texttt{rdrand} with the reservation of 8 unused SSE registers. 
This method enables the compiler to either allocate these registers freely, as in Variant \texttt{rdrand}, or reserve them for storing secrets from sensitive memory access instructions, similar to Strategy 2.

\begin{table}[htbp]
\centering
\caption{Profit analysis of Strategy 2.}
\label{tab:variant3benefit}
\scriptsize
\begin{tabular}{ccccc}
\hline
    \multirow{2}{*}{\textbf{Implementation}}
    &\multicolumn{1}{c}{\textbf{V-RDRAND}}
    &\multicolumn{2}{c}{\textbf{RSV-SSE}}
    &\multicolumn{1}{c}{\textbf{Strategy 2}}\\
    \cline{3-4}
    & \textbf{Factor} & \textbf{Cycles} & \textbf{Factor} & \textbf{Factor} \\
\hline
    libsodium-EdDSA &3.39 &680723	&3.55 &1.59 \\
    libsodium-SHA512 &2.38 &121020	&2.85 &2.47 \\
\hline
    mbedTLS-AES &1.21 &1135500	&2.39 &2.09 \\
    mbedTLS-Base64 &1.50 &20820	&1.86 &1.38 \\
    mbedTLS-CC20 &1.26 &1329720	&2.23 &2.10 \\
    mbedTLS-ECDH &1.81 &7981500	&1.88 &1.46 \\
    mbedTLS-ECDSA &1.85 &10137960	&2.57 &1.38 \\
    mbedTLS-RSA &2.29 &5459100	&2.90 &1.52 \\
\hline
    OpenSSL-ECDH &6.17 &6121680	&6.21 &1.54 \\
    OpenSSL-ECDSA &3.96 &73640040	&4.37 &2.04 \\
\hline
    WolfSSL-AES &2.10 &1916760	&3.00 &2.97 \\
    WolfSSL-CC20 &1.09 &614670	&1.28 &1.12 \\
    WolfSSL-ECDH &2.12 &849300	&2.44 &1.50 \\
    WolfSSL-ECDSA &4.96 &21068010	&4.98 &2.16 \\
    WolfSSL-EdDSA &3.14 &1428990	&3.41 &1.53 \\
    WolfSSL-RSA &1.24 &902452	&1.83 &1.28 \\
\hline
    Average &2.53 &- &2.98 &1.76 \\
\hline
\end{tabular}
\end{table}

\T~\ref{tab:variant3benefit} presents the results of the new ``RSV-SSE'' strategy, which demonstrates an average overhead of 2.98$\times$, higher than that of both Variant \texttt{rdrand} and Strategy 2. 
Certain cryptographic functions (SHA512 from libsodium, AES and ChaCha20 from mbedTLS, AES, ChaCha20 and RSA from WolfSSL) show decreased performance when SSE registers are used to safeguard sensitive memory accesses, negatively impacting other program components. 
Conversely, functions that benefit from this strategy display significant performance improvements, particularly ECDSA from OpenSSL, where Strategy 2 shows exceptional profitability.

We further investigate the allocation of SSE registers for sensitive memory accesses in the ECDSA implementation from OpenSSL. By adhering to the secret-aware register allocation principles outlined in \S~\ref{subsec:registeralloc}, we analyze the range of MBBs in which these registers are assigned.
The findings indicate a significant advantage when registers are utilized to store stack memory within loops, leading to reduced cycle counts compared to using masking for protection, as seen in Variant \texttt{rdrand}. 
However, if protection extends to insufficient sensitive stack memory within a loop, such as allocating SSE registers for sequential sensitive memory accesses, the benefits may not outweigh the drawbacks of not using these registers in other areas of the cryptographic program. 

\parh{Obfuscation reducing memory instructions.}
The performance improvement achieved through obfuscation, compared to register-based protection, primarily stems from the reduction in memory load operations for currently used nonces. 
Our odd-even selection strategy ensures that each heap unit undergoes two memory write operations, one for inserting the true secret and another for storing the decoy value (with varying nonces) into either the original heap position or the secure buffer. 
However, each corresponding load operation requires only a single memory access, avoiding the additional memory loads needed in Strategy 2 to retrieve the currently used nonce from its buffer. 
As a result, Strategy 3 significantly reduces the total number of memory instructions, leading to improved performance compared to approaches that rely on complex nonce management.



\begin{table*}[t]
\centering
\caption{Performance comparison with \ftool\ based on the same number of tainted functions. The replication of \ftool\ is conducted on its \textsc{Fast} version. To quantify the additional average cycles incurred by the patched cryptography library, we introduce a metric labeled ``consume'', which shows an average consumption of 3,661 cycles for \tool\ compared to a significant 3.46$\times$ increase for \ftool.}
\label{tab:comparisonftool}
\footnotesize
\resizebox{\linewidth}{!}{
\begin{tabular}{ccccc|cccc|cccc}
\hline
    \multirow{2}{*}{\textbf{Implementation}}
    &\multicolumn{4}{c|}{\textbf{\ftool\quad \tool}}
    &\multicolumn{4}{c|}{\textbf{\ftool-\textsc{Fast}}}
    &\multicolumn{4}{c}{\textbf{\tool-Variant 2}}\\
    \cline{2-13}
    & \textbf{FUN} & \textbf{INS} & \textbf{FUN} & \textbf{INS}
    & \textbf{Orig} & \textbf{AES} & \textbf{Factor} & \textbf{Consume}
    & \textbf{INS} & \textbf{AES} & \textbf{Factor} & \textbf{Consume}\\
\hline
    libsodium-EdDSA &14 &616 &17 &1311 & 131790 &779580 &5.92 & 1052 & 630 &271557	&1.42	&127 \\
    libsodium-SHA512 &6 &155 &6 &586 & 60060 &103920 &1.73 & 283 & 211 &84139	&1.98	&197 \\
\hline
    mbedTLS-AES &19 &96 &17 &326 & 312990 &1165470 &3.72 & 8880 & 318 &780317	&1.64	&960 \\
    mbedTLS-Base64 &5 &25 &9 &386 & 10230 &16440 &1.61 & 248 & 26 &14897	&1.33	&141 \\
    mbedTLS-ChaCha20 &15 &234 &14 &275 & 397080 &922800 &2.32 & 2247 & 267 &1193730	&2.00	&2235 \\
    mbedTLS-ECDH &20 &65 &51 &1063 & 3314160 &10658700 &3.22 & 112993 & 171 &7760267	&1.82	&20505 \\
    mbedTLS-ECDSA &51 &448 &69 &1446 & 1385790 &10117620 &7.30 & 19491 & 1156 &7179808	&1.82	&2801 \\
    mbedTLS-RSA &35 &300 &44 &1238 & 2893260 &11347680 &3.92 & 28181 & 673 &3640590	&1.93	&2609 \\
\hline
    OpenSSL-ECDH &11 &117 &721 &12876 & 373800 &583680 &1.56 & 1794 & 123 &1280850	&1.30	&2395 \\
    OpenSSL-ECDSA &91 &653 &796 &17834 & 3236490 &6735210 &2.08 & 5358 & 991 &22131750	&1.31	&5325 \\
\hline
    WolfSSL-AES &5 &204 &6 &499 & 401550 &725040 &1.81 & 1586 & 192 &809066	&1.27	&885 \\
    WolfSSL-ChaCha20 &9 &182 &14 &404 & 706740 &1022880 &1.45 & 1737 & 343 &505800	&1.05	&70 \\
    WolfSSL-ECDH &10 &291 &20 &543 & 191280 &385440 &2.02 & 667 & 147 &701408	&2.01	&2399 \\
    WolfSSL-ECDSA &40 &328 &59 &932 & 3143010 &8558370 &2.72 & 16510 & 242 &8321686	&1.97	&16920 \\
    WolfSSL-EdDSA &31 &835 &28 &715 & 312360 &746250 &2.39 & 520 & 616 &893040	&2.13	&769 \\
    WolfSSL-RSA &48 &696 &43 &704 & 537750 &1200660 &2.23 & 952 & 529 &615510	&1.25	&234 \\
\hline
    Average &26 &328 &120 &2571 & - &- &2.87 &12656 & 415 &-	&1.64	&3661 \\
\hline
\end{tabular}}
\end{table*}

\subsection{Compared with \ftool}
\label{subsec:comparison}

\noindent \textbf{Dynamic Taint Analysis.}
Identifying sensitive memory access instructions in each cryptography library takes less than 20 minutes. 
Specifically, taint analysis for symmetric primitives like AES and ChaCha20 is completed in just 2 to 3 minutes. 
Although asymmetric primitives such as RSA and ECDSA, which use blinding and nonces, require more time, the taint analysis still finishes within approximately 10 minutes, which is considered a reasonable duration.

In the evaluation shown in \T~\ref{tab:comparisonftool}, we systematically document the number of functions and instructions identified by taint analysis for each cryptography library, as reflected in the fourth and fifth columns. 
The results reveal a wide range, with a minimum of 6 sub-functions in SHA512 from libsodium and a maximum of 796 sub-functions in ECDSA from OpenSSL, averaging 120 functions across the analyzed implementations. 
The tainted instructions within these functions vary significantly, from 275 in ChaCha20 of mbedTLS to 17,834 instructions in ECDSA of OpenSSL. 
In comparison to \ftool, as illustrated in the second and third columns of \T~\ref{tab:comparisonftool}, our approach demonstrates more comprehensive and detailed tracking of sensitive data.

\noindent \textbf{Mitigation Coverage Comparison.}
Accurate identification of tainted instructions is essential, as these instructions are key candidates for the subsequent mitigation process. 
Although \tool\ and \ftool\ use different compilers, Clang and GCC, both maintain consistent control flow when subjected to the same optimization settings. 
Capitalizing on this similarity, we extract binaries with instrumentation information from the taint analysis stage and analyze the basic blocks of tainted functions by constructing control-flow graphs (CFGs). 
For instance, we compare the CFGs of the function \texttt{SHA512\_Transform} from libsodium-SHA512 in \F~\ref{fig:cfg}. 
The CFGs exhibit notable structural similarities, and we manually identify three key nodes to further refine the control flow. 
A line-by-line comparison of the code within both binaries reveals only minor differences, averaging 2\%. 
Thus, we conclude that the variance in compilers does not significantly affect the discrepancy in the number of tainted instructions identified by \tool\ and \ftool\ as shown in \T~\ref{tab:comparisonftool}.

\begin{figure}[htbp]
\centering
\includegraphics[width=0.95\linewidth]{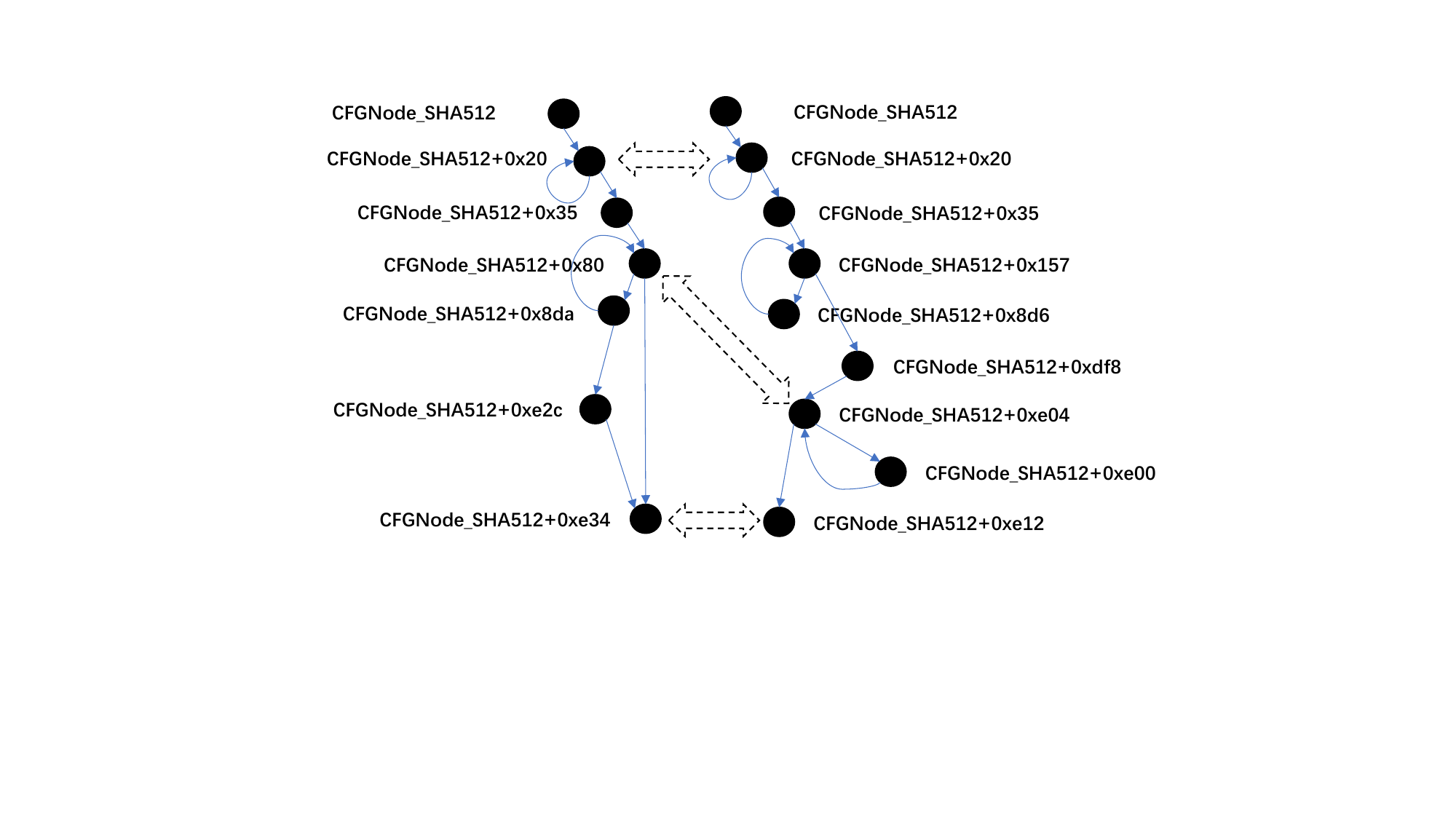}
\caption{Function \texttt{SHA512\_Transform} from the libsodium-SHA512 serves as an example to illustrate the construction of CFGs and identify critical nodes.}
\label{fig:cfg}
\end{figure}

Our analysis reveals that the cryptography libraries listed in \T~\ref{tab:comparisonftool} are monitored for a similar number of tainted functions, with the exception of two implementations from the OpenSSL libraries. 
To investigate the variance in taint analysis among these libraries, we randomly select five tainted functions for comparison. 
We find that the additional instructions identified by \tool\ can be categorized into several types.
Notably, \tool\ conservatively tags the parameters of temporary variables used to store intermediate results, such as loop variables involved in sensitive calculations. 
In contrast, \ftool\ overlooks these intermediate and resultant variables, focusing primarily on directly involved secret keys.
Consequently, in the OpenSSL libraries, \tool\ identifies not only the parameters of temporary variables but also intermediate and resultant variables, resulting in more comprehensive identification of relevant functions and instructions compared to \ftool, as illustrated in \T~\ref{tab:comparisonftool}.

\noindent \textbf{Performance Comparison.}
To ensure a fair performance comparison between \tool\ and \ftool, we evaluate both tools using as many common tainted instructions as possible within the same cryptography library. 
However, achieving this is challenging due to the variability in the number of tainted functions and instructions identified by each tool. For instance, protecting tainted instructions within loops can lead to increased execution cycles. 
To mitigate this issue, we manually select tainted functions that are common to both tools, allowing \tool\ to apply its mitigation process to functions flagged by \ftool. 
This approach helps align the execution flow between the two tools. Nevertheless, \tool\ still identifies more instructions than \ftool, although we strive to make the comparison as close as possible.

The details of selected tainted instructions are presented in the fourth-to-last columns of \T~\ref{tab:comparisonftool}. 
Both \tool\ and \ftool\ utilize the \texttt{vaesenc} instruction from AES-NI to insert masking instructions, thereby eliminating the influence of any optimization specific to \tool. 
The results indicate that \tool\ achieves an average overhead of 1.64$\times$, whereas \ftool\ experiences a higher overhead of 2.87$\times$. 
We introduce the metric ``consume'' to quantify the average additional cycles, with \tool\ incurring 3,661 cycles compared to a 3.46$\times$ increase observed for \ftool.

The performance differences between \tool\ and \ftool\ are largely attributed to specific characteristics of \ftool\ that involve frequent jumps to instrumentation code, monitoring of \texttt{malloc} for heap memory allocation of nonce buffers, and the runtime initialization of the nonce cache. 
In contrast, \tool\ employs a compiler-aided approach that optimizes these time-intensive tasks by sequentially inserting mask instructions and efficiently managing nonce buffers in the \texttt{.bss} section. 
This design choice significantly enhances the mitigation process, reducing associated overhead and improving overall performance.

\subsection{Mitigation Cases}
\label{app_cases}

\parh{ossl\_ec\_scalar\_mul\_ladder.}
Using the ECDSA Montgomery ladder algorithm from OpenSSL as an example (\F~\ref{fig:channel}(a)), we examine how secrets are protected by masking operations in Strategy 1 of \tool\ with three different random number generation techniques.
The effectiveness of our software-based probabilistic encryption scheme relies heavily on the quality of random number generation, as the reliability and randomness of nonces are crucial. 
Without protection, a vulnerability at the point where $pbit \leftarrow pbit \wedge kbit$ spills the secret $pbit$ into memory, allowing attackers to infer each bit of the secret by mapping ciphertext-plaintext pairs.
To assess the impact of our masking method, we disassemble the binary to locate the vulnerability, then use Pintool~\cite{luk2005pin} to observe the execution context at this point.
By running the cryptography library twice with different ECDSA private keys, we track the writing of $pbit$ 512 times within the loop. 
We then measure the entropy of the collected secret sequence as $\mathit{H(X)} = - \sum_{x \in \mathcal{X}} p(x) log_{2} p(x)$ to quantify the distribution changes and evaluate the scheme’s effectiveness.

\begin{figure}[htbp]
\centering
\includegraphics[width=0.95\linewidth]{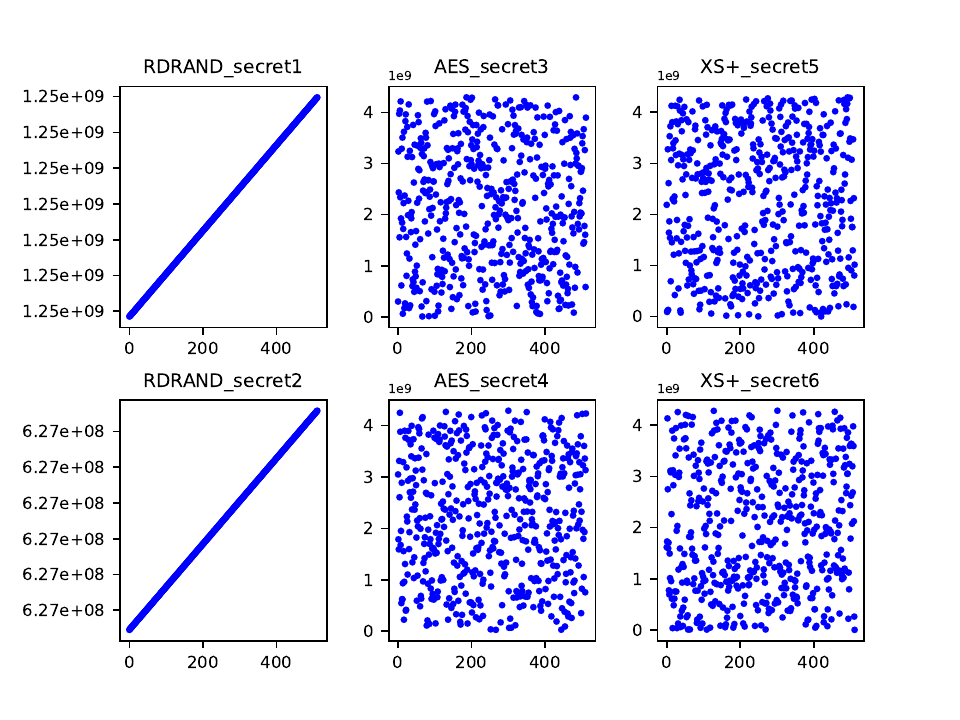}
\caption{Scatter distribution of masked $pbit$ under different variants. Each secret sequence comprises 512 values.}
\label{fig:entropy}
\end{figure}

The entropy analysis of six different unprotected secrets (0.99946, 0.99911, 0.99972, 0.99910, 0.99841, and 0.99814) shows values around 0.99 for each, indicating a significant leakage of information, as each bit can only be 0 or 1. 
This low entropy allows attackers to easily infer patterns in the secret sequences via ciphertext side-channel attacks. 
However, after applying Strategy 1 of \tool, the entropy of the secret sequences increases to approximately 9.0, close to the theoretical maximum for 512 random numbers. 
This higher entropy reflects that the secrets become random and unpredictable, making it extremely difficult for attackers to deduce any patterns. 
The scatter distribution of these masked secret sequences in \F~\ref{fig:entropy} further confirms the effectiveness of \tool's masking approach in preventing leakage.


\parh{BN\_constant\_swap.}
Similarly, we illustrate how \tool\ mitigates the vulnerable constant-time swap function \texttt{BN\_constant\_swap} (\F~\ref{fig:channel}(b)) using diversion-based obfuscation. 
This function takes two variables, $a$ and $b$, along with a secret-dependent decision $C$ (e.g., $kbit$ in line 4 of \F~\ref{fig:channel}(a)).
Without mitigation, if $C = 1$, the values of $a$ and $b$ are exchanged, producing observable changes in the ciphertext; if $C = 0$, the ciphertext remains unchanged.

Under Strategy 3 of \tool, isolated secure buffer space is allocated for both $a$ and $b$, which either accepts the original value or a decoy (e.g., a varying nonce). 
In this mechanism, the secure buffer always changes. For the original heap space, it remains unchanged under \textbf{Obfuscation 2}, or is modified under \textbf{Obfuscation 1}, where the original value is transferred to the secure buffer and replaced with a decoy. 
From an attacker’s perspective, multiple locations change unpredictably, preventing any inference about whether $a$ and $b$ undergo a constant-time swap.

\begin{figure}[htbp]
\centering
\includegraphics[width=\linewidth]{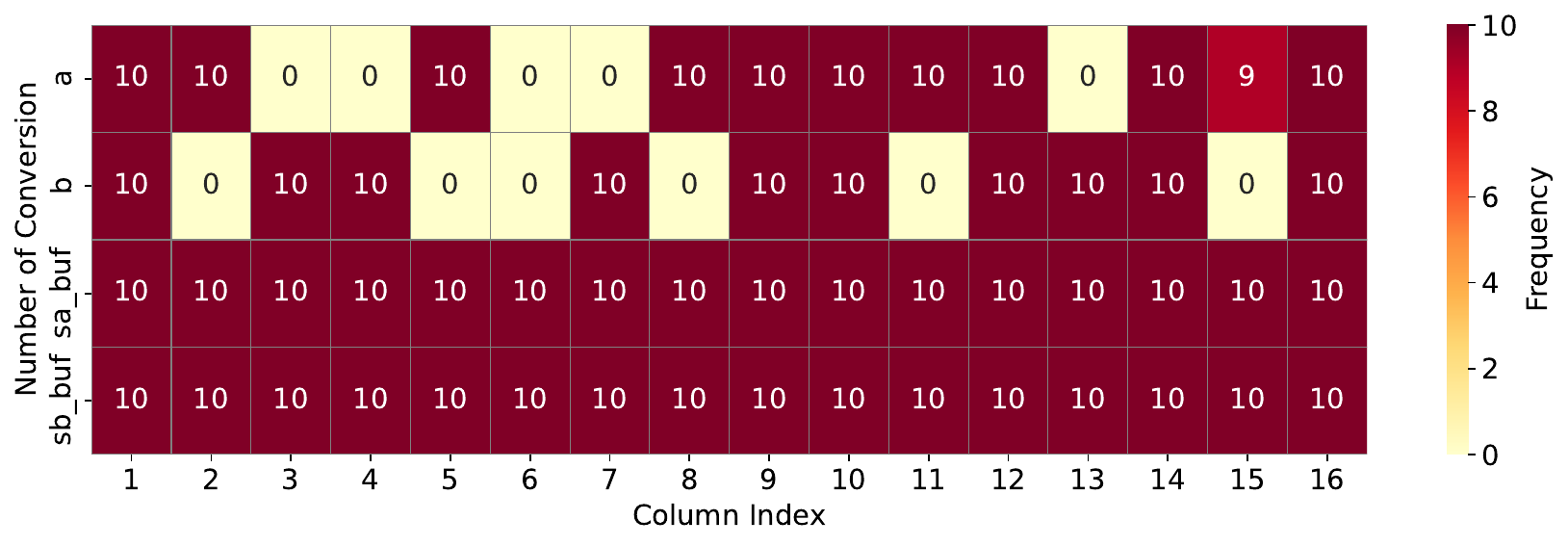}
\caption{Each cell shows the number of value updates over 10 runs for both the original and secure buffer locations. The uniform and patternless distribution demonstrates that the actual storage of swapped variables cannot be inferred by an attacker.}
\label{fig:heatmap}
\end{figure}

We apply Strategy 3 over 10 runs to obfuscate the big number variables in \texttt{BN\_constant\_swap} with varying starting addresses (each constant-time swap involves 512 iterations). 
A heatmap is used to record the number of changes at both the original locations and the secure buffer, as displayed in \F~\ref{fig:heatmap}. 
For illustration, we show the first 16 cells in the figure. 
The heatmap reveals a memory-change pattern that is completely uniform and unpredictable, preventing an attacker from inferring the actual storage locations of the swapped variables.

\section{Discussion}
\label{sec:discussion}

\subsection{Alternative Methods}
\label{subsec:padding}

\ztool~\cite{patschke2025zebrafix} splits plaintext into 8-byte chunks and pads each with an 8-byte write counter using an interleaving method to enhance freshness.
However, this method requires considering the memory structure of the interleaved variables to implement appropriate instrumentation.
\ztool\ leverages variable type information from the compilation stage and employs pointer analysis to determine memory structures for more efficient space interleaving.
For heap memory allocations that lack explicit variable type information, however, it adopts a worst-case layout, storing only one byte of useful data in every 16 bytes, thereby increasing memory overhead by 16$\times$.
Moreover, since \ztool modifies variable memory layouts, space-optimization techniques may become ineffective.
In worse cases, it can even cause incorrect content retrieval when accessing memory via indirect pointers and offsets, thereby impacting both compatibility and correctness.

\parh{Our Memory Consumption.}
The runtime stack memory consumption of \tool remains at the same level as \ftool.
For the primary memory consumption of \tool for heap space, it arises from three dedicated buffers. 
Specifically, a pre-generated nonce buffer (1 KB) is allocated to store a large batch of random nonces in advance. In addition, a 2 MB buffer is reserved for the currently used nonces (excluding potential hash collisions), ensuring efficient access during runtime.
To further assist obfuscation, a secure buffer of 2 MB is allocated as well.

\subsection{Potential Application}
\label{subsec:app}

\parh{Extensions and Compatibility.}
\tool\ offers a flexible interface for integrating support for instructions not yet covered, highlighting its extensibility for future security updates and architectural changes.
\tool\ is designed to be versatile and independent of specific compiler frameworks. 
When using GCC, the LLVM-specific syntax can be replaced with GCC IR while keeping the same taint propagation rules. 
The tainted instructions identified during analysis are targeted for protection within the GCC IR. 
Integrating the mitigation process into GCC's register allocation phase is straightforward, reflecting \tool's design for easy adaptation to different compiler pipelines.

\parh{Impacts on Real-world Applications.}
It is important to contextualize the performance overhead measured in our evaluation. Our analysis focuses on isolated cryptographic primitives, which amplifies the perceived slowdown. In real-world applications, however, these cryptographic operations often constitute only a small fraction of the total execution time. For instance, in a typical TLS handshake, the asymmetric cryptography that requires hardening against ciphertext side-channel attacks is used only during the initial negotiation phase. The bulk of the protocol’s runtime is dedicated to transmitting application data, which does not require such costly protections. To illustrate, cryptographic computations might account for only \texttt{5–10}\% of the total time in a complex application handling network I/O and business logic. Therefore, even a 2$\times$ overhead on the cryptographic portion would translate to a much smaller, and often negligible, impact of about \texttt{5–10}\% on the application’s overall performance. This suggests that the practical impact of \tool\ in real-world deployments, such as TLS servers or secure communication endpoints, would be significantly lower than the overheads reported for the primitives in isolation.

\subsection{Limitations}
\label{subsec:limit}

\parh{Recompilation.}  
Recompilation is inherent to compiler-based protection methods. 
In our case, taint information from the initial compilation is analyzed to guide subsequent repair with a program-wide perspective. 
We observe that the additional compilation overhead in the second pass is negligible.

\parh{Inline Assembly.}
Inline or pure assembly, which is commonly employed in cryptographic software for performance-critical routines, presents a significant challenge for compiler-based taint analysis. 
Unlike high-level code, assembly bypasses many of the abstractions and analysis hooks available to the compiler, making it difficult to track the flow of sensitive data accurately. 
Assembly code usually forms a small, stable portion of cryptographic software (e.g., AES) and can be manually adapted to integrate with the compiler’s taint tracking and mitigation pipeline with minimal maintenance overhead.
While this limitation is intrinsic to all compiler-based approaches, it does not detract from their broader effectiveness. 

\parh{Multi-threading.}  
We do not consider multi-threading in this work. Supporting multi-threaded execution would require additional locking operations, inevitably incurring performance overhead. 
Nevertheless, our design achieves the same level of thread-safety as \ftool, ensuring comparable security guarantees in this aspect.

\parh{Dynamic Taint Analysis.}
\tool\ uses dynamic taint analysis to identify potential sensitive memory access instructions and applies conservatively across multiple secrets and inputs. 
However, it can also lead to unnecessary protection of memory writes that don't actually cause ciphertext side-channel leakage, which adds runtime overhead. 
Reducing this overhead would require a precise, static whole-program taint analysis that can accurately identify true leakage points. 
Unfortunately, this remains a challenge, as current research only allows static analysis of a small portion of the program~\cite{wang2019identifying, brotzman2019casym}.

While \tool\ significantly reduces the exposure to ciphertext side-channel attacks, potential residual risks may include unexecuted paths or subtle metadata leaks.
Overall, while \tool\ provides robust security against ciphertext side channels, absolute security is challenging to guarantee, and the remaining attack surface is reduced to areas that require attacks combining multiple vectors or exploiting vulnerabilities beyond the scope of current mitigation mechanisms.



\section{Conclusions}
\label{sec:conclusion}


\tool\ demonstrates that compiler-based transformations can systematically defend ciphertext side channels in TEEs.
Specifically, \tool\ achieves path-extended coverage beyond binary repair.
Its three strategies eliminate ciphertext, access patterns, and timing leakages while preserving program semantics.
Formal small-step proofs establish semantic equivalence between original and transformed programs, with compiler-issued certificates ensuring atomicity, register integrity, and dummy isolation.
\tool\ thus provides verifiable, high-coverage, and low-overhead protection, offering stronger guarantees than post-compilation masking.

\clearpage
\bibliographystyle{plain}
\bibliography{reference}

\end{document}